\documentclass[sigconf]{aamas} 

\usepackage{balance} 

\makeatletter
\gdef\@copyrightpermission{
  \begin{minipage}{0.2\columnwidth}
   \href{https://creativecommons.org/licenses/by/4.0/}{\includegraphics[width=0.90\textwidth]{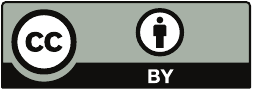}}
  \end{minipage}\hfill
  \begin{minipage}{0.8\columnwidth}
   \href{https://creativecommons.org/licenses/by/4.0/}{This work is licensed under a Creative Commons Attribution International 4.0 License.}
  \end{minipage}
  \vspace{5pt}
}
\makeatother

\setcopyright{ifaamas}
\acmConference[AAMAS '25]{Proc.\@ of the 24th International Conference
on Autonomous Agents and Multiagent Systems (AAMAS 2025)}{May 19 -- 23, 2025}
{Detroit, Michigan, USA}{Y.~Vorobeychik, S.~Das, A.~Nowé  (eds.)}
\copyrightyear{2025}
\acmYear{2025}
\acmDOI{}
\acmPrice{}
\acmISBN{}

\usepackage[utf8]{inputenc}
\usepackage[T1]{fontenc}

\usepackage{dsfont}
\usepackage[T1]{fontenc}
\usepackage{mathtools}
\usepackage{natbib}
\usepackage{amsmath,amsfonts,amsthm,textcomp,enumitem,thm-restate}
\usepackage[capitalize]{cleveref}

\usepackage{tikz}
\usetikzlibrary{positioning}

\usepackage{pgfplots}
\pgfplotsset{width=\columnwidth,height=0.516*\columnwidth,compat=1.9}

\makeatletter
\theoremstyle{plain}
\theoremstyle{acmplain}
\@ifundefined{observation}{%
\newtheorem{observation}{Observation}[section]
}{}
\@ifundefined{Description}{%
	\newcommand{\Description}[2][foo]{}{}
}{}
\@ifundefined{theorem}{%
\newtheorem{theorem}{Theorem}[section]
}{}
\@ifundefined{conjecture}{%

}{}
\@ifundefined{proposition}{%

}{}
\@ifundefined{lemma}{%
\newtheorem{lemma}[theorem]{Lemma}
}{}
\@ifundefined{corollary}{%
\newtheorem{corollary}[theorem]{Corollary}
}{}
\theoremstyle{definition}
\theoremstyle{acmdefinition}
\@ifundefined{example}{%

}{}
\@ifundefined{definition}{%
\newtheorem{definition}[theorem]{Definition}
}{}
\theoremstyle{plain}
\theoremstyle{acmplain}
\makeatother

\usepackage[capitalize]{cleveref}

\DeclareMathOperator*{\E}{\mathbb{E}}
\def\eps{\varepsilon}
\def\problem{\textsc}
\def\problemClass{\textsf}
\def\np{\problemClass{NP}}
\def\apx{\problemClass{APX}}
\def\p{\problemClass{P}}
\def\absolute#1{\left\lvert #1 \right\rvert }
\def\asg{\mathcal{A}}
\def\deq{\mathbin{:=}}
\DeclareMathOperator{\Berd}{Ber}
\def\berd{\fce{\Berd}}
\def\fce#1#2{#1\!\left(#2\right)}
\def\fceb#1#2{#1\!\left[ #2 \right]}
\def\suchthat{\;\vert\;}
\def\pr#1{\fceb{\Pr}{#1}}
\def\prt#1{\fceb{\Pr}{\text{#1}}}
\def\Acknowledgments{
Part of this work was carried out while Amanda Wang and Filip Úradník were participants in the DIMACS REU program at Rutgers University. Amanda Wang was supported by the NSF award CNS-2150186, the REU supplement to NSF 2208663, and by Princeton University's Office of Undergraduate Research Undergraduate Fund for Academic Conferences through the Hewlett Foundation Fund.
Filip Úradník was supported by CoSP, a project funded by European Union’s Horizon 2020 research and innovation programme, grant agreement No. 823748, by the Czech Science Foundation Grant 25-15714S, and by the Charles University Grant Agency (GAUK 206523).
Jie Gao would like to acknowledge funding through NSF IIS-2229876, DMS-2220271, DMS-2311064, CCF-2208663,  CCF-2118953.
The authors would also want to thank Kevin Lu, Júlia Križanová and Rhett Olson for their useful discussion and insightful comments.
}
\def\Abstract{
Sequential learning models situations where agents predict a ground truth in sequence, by using their private, noisy measurements, and the predictions of agents who came earlier in the sequence.
We study sequential learning in a social network, where agents only see the actions of the previous agents in their own neighborhood.
The fraction of agents who predict the ground truth correctly depends heavily on both the network topology and the ordering in which the predictions are made. 
A natural question is to find an ordering, with a given network, to maximize the (expected) number of agents who predict the ground truth correctly.
In this paper, we show that it is in fact \np-hard to answer this question for a general network, with both the Bayesian learning model and a simple majority rule model. 
Finally, we show that even approximating the answer is hard.
}
\def\Keywords{Social Networks, Social Learning, Sequential Learning, Bayesian Learning, Complexity, Majority Vote}

\def\netlearn{\problem{Network Learning}}
\def\netlearnopt{\problem{Opt Network Learning}}
\def\maxsat{\problem{Max \sat}}
\def\sat{\problem{3-SAT}}
\def\network{\mathcal{N}}
\def\lr{\overline{\mathcal{L}}}
\def\clr{\mathcal{L}}
\def\olr{\lr^*}
\def\oclr{\clr^*}
\def\cell#1{\fce{\mathbb{C}}{#1}}
\def\gadget#1{\fce{\mathbb{G}}{#1}}
\def\vars{\chi}

\def\ourparagraph{\paragraph*}
\newcommand{\figThickness}{thick}
\newcommand{\figCircSize}{20}
\newcommand{\figSamples}{300}
 
\def\pone{\clr_1}
\def\ptwo{\clr_2}
\def\pthree{\clr_3}
\def\pzero{\clr_0}
\def\pcell{\clr_{\text{cell}}}

\newcommand{\appbayes}{\Cref{app:bayesian_gadgets}}
\newcommand{\appmaj}{\Cref{app:majority_proof}}

\acmSubmissionID{564}

\title{Maximizing Truth Learning in a Social Network is NP-hard}

\author{Filip \'Uradn\'ik}
\authornote{These authors contributed equally.}
\affiliation{
  \institution{Charles University}
  \city{Prague}
  \country{Czech Republic}}
\email{uradnik@kam.mff.cuni.cz}

\author{Amanda Wang}
\authornotemark[1]
\affiliation{
  \institution{Princeton University}
  \city{Princeton, NJ}
  \country{United States}}
\email{aw4309@princeton.edu}

\author{Jie Gao}
\affiliation{
  \institution{Rutgers University}
  \city{Piscataway, NJ}
  \country{United States}}
\email{jg1555@rutgers.edu}

\begin{abstract}
  \Abstract{}
\end{abstract}


\newcommand{\BibTeX}{\rm B\kern-.05em{\sc i\kern-.025em b}\kern-.08em\TeX}

\begin{document}


\pagestyle{fancy}
\fancyhead{}


\maketitle 


\section{Introduction}

Information acquisition, opinion formation and decision making are deeply embedded in a social context.
There are many situations in which people make decisions using information that carries uncertainties, and such decisions are easily influenced by the decisions of others. It is therefore of paramount importance to understand how to effectively use information of inherent uncertainty, while considering how social exchanges can reduce such uncertainties.

Broadly, there are two primary families of models that capture decision making processes in a social network. The first family considers opinion dynamics~\cite{Mossel2017-sd, Acemoglu2011-tx}, where all agents have individual opinions that are repeatedly updated based on information exchange with others. Opinion dynamics studies the evolution of an opinion landscape over time and asks whether, and how quickly, a consensus can be obtained. Recent work also studies the lack of consensus (i.e., polarization) and asks whether this can be modeled and explained with natural factors~\cite{Wang2022-uy}. The second family considers sequential decision making processes~\cite{Golub2017-qo,Mobius2014-oy}, where agents make one-shot decisions. It is also assumed that there is an unknown ground truth state which all agents wish to learn, whereas opinion dynamics often omits such an assumption. Therefore, this second setting is termed sequential social `learning'. 

\ourparagraph{Our Model.}
We consider the classical sequential learning model, where $n$ agents sequentially predict an unknown ground truth $\theta$~\cite{Jadbabaie2012-ob,Mossel2014-mv}. Each agent has an independent private measurement of $\theta$. We consider the `bounded belief' setting, where each agents' private measurement has the same probability $p$, a constant away from 1, of being correct~\cite{Smith2000-wk,Acemoglu2011-vj}.
In addition, each agent has access to the predictions of agents earlier in the sequence. Ideally, agents can use the information extracted from the earlier predictions to improve their own prediction. However, a well-known problem that arises is \emph{information cascade} or \emph{herding}~\cite{Banerjee1992-ra,Bikhchandani1992-rs,Welch1992-yt,Smith2000-wk,Chamley2004-or}, where sufficiently many wrong predictions early on can trigger all subsequent agents to ignore their own private signals and `follow the herd'. Notably, this can occur even for fully rational agents, in the absence of behavioral factors like peer pressure. So herding arises not by fault of the agents, but as a result of the sequential structure of the setting. Indeed, whenever a large enough part of the crowd discards its private information, whether rationally or not, the crowd as a whole is unable to learn the ground truth. 

Motivated by the problem of information cascades, a lot of follow-up work examines how to restore truth learning in crowds.
One approach is to limit the visibility of agents~\cite{Smith1991-sy,Sgroi2002-rz,Acemoglu2011-vj} so that incorrect predictions do not propagate. A natural setup is to consider a social network, rather than an unstructured crowd, in which an agent can only see the actions of its neighbors earlier in the sequence~\cite{Bahar2020-am,arieli2020social,lu24enabling}.
It can be shown that certain network structures, coupled with a good agent decision ordering, can in fact achieve a strong learning result called asymptotic network-wide truth learning~\cite{lu24enabling}: that all $n$ agents, except $o(n)$ of them, successfully predict the ground truth as $n$ goes to infinity. In other words, the average network learning rate, or the average success probability over all agents, approaches $1$. 

One natural open problem from~\cite{lu24enabling} asks whether given a social network, we can either find a good decision ordering or decide if a good ordering exists. The authors in~\cite{lu24enabling} presented sufficient conditions and impossibility results, but the big picture is still largely unclear. In general, there are two factors that prohibit truth learning. If the network is too sparse or a constant fraction of agents make  decisions using only their private signals, then their success probability is bounded by a constant away from $1$, and network-wide truth learning is already doomed. On the other hand, if the network is too dense or almost every agent is well connected with agents earlier in the ordering, then herding happens. Again in this case, truth learning is not possible. Regarding the decision problem, ~\cite{lu24enabling} conjectures that deciding if a network admits an ordering enabling asymptotic truth learning is \np-hard. 

\ourparagraph{Our  Results.} This work focuses on the problem of deciding if the best possible average learning rate in a network exceeds a given threshold $1-\varepsilon$. We prove that this problem is NP-hard both for networks of fully rational agents and those with bounded rationality. Intuitively, we expect this problem to be hard---naively, there are exponentially many decision orderings to check---but proving it formally is highly non-trivial. At a cursory glance, the two barriers for truth learning suggest that a successful network should avoid very sparse and very dense structures. So a natural approach may be to relate network learning and the maximum independent set or maximum clique decision problems. However, it remains unclear how to characterize a network's learning rate by the sizes of its max independent set or its max clique. Thus, it is not obvious how to reduce either problem to the network learning problem. 

We instead use a reduction from \sat, the canonical \np-hard decision problem asking if a given 3-CNF formula is satisfiable.
We construct a network from an input 3-CNF instance and map decision orderings on the network to boolean variable assignments.
In our construction, satisfied clauses correspond to subgraphs with higher learning rates than unsatisfied ones.
Therefore, the larger the number of satisfied clauses, the higher the network learning rate.
We prove that all orderings corresponding to a satisfying assignment achieve a strictly higher network learning rate than those corresponding to non-satisfying assignments.
Thus, deciding if the optimal learning rate exceeds a well-chosen threshold immediately implies an answer to \sat{}.
While the high-level idea is clean, the details are fairly technical due to the dependence between predictions of neighboring agents.
We defer the most technical parts of the proof to the Appendix.

Next, we focus on the approximation hardness of this problem.
We construct a reduction to \maxsat{}, which asks for the maximum number of satisfiable clauses under any variable assignment in a \sat{} instance.
It is known~\cite{Hastad2001-fg} that computing a solution that satisfies more than $\frac 78 M^*$ clauses is \np-hard, where $M^*$ is the maximum number of clauses that can be satisfied.
Using the ideas from the \sat{} reduction, we prove that it is impossible to find an efficient approximation up to an arbitrary constant, unless $\p=\np$.
Specifically, this also allows us to strengthen our previous claim about \np-hardness---that finding the optimal learning rate is \np-hard for any fixed prior $ p \in (\sqrt{7/8}, 1) $.

As noted above, we present hardness for sequential social learning both when agents are fully rational and when they have bounded rationality. Fully rational agents can be realistic, for instance, in modelling financial traders, while agents with bounded rationality may be more faithful models for voters or users on a social platform. More concretely, fully rational agents predict via Bayesian inference, whereas agents with bounded rationality use a simpler heuristic, such as majority vote over all available signals. 

From a practical perspective, Bayesian inference is computationally expensive to implement in simulations, and all prior work on this topic~\cite{Bahar2020-am,arieli2020social,lu24enabling} used bounded rationality in experiments, such as majority vote. However, while simpler for implementation, majority vote is actually harder for analysis. With the Bayesian model, the success rate of an agent is at least as high as the highest success rate of its earlier neighbors, since a node can never do worse than copying the action of an earlier neighbor. This monotonicity can be helpful for constructing a good ordering. However, with majority vote, this nice property no longer holds, so the validity of the hardness reduction needs to be reconsidered for the majority vote model. We modify the construction and restore the same hardness claims for the majority vote model. 
 
\ourparagraph{Related Literature.}
A number of recent works consider models with repeated observations and information exchange, asking if the agents successfully learn the ground truth 
(see e.g., \cite{Jadbabaie2012-ob,Mossel2014-mv}). One major model choice is how an agent aggregates information from available signals in the network. The most natural choice is the Bayesian model, in which agents compute the posterior probability for $\theta$ using all available information and any common knowledge, such as the network topology. 
A recent line of work showed that computing an agent's prediction via Bayesian inference with repeated opinion exchange is PSPACE-hard~\cite{hazla2019reasoning,Hazla2021-vf}. It is one of the few works, to the best of our knowledge, that formally characterizes the complexity of decision making in a social network. 
From this perspective, our work adds to the relatively scarce literature that combinatorial complexity arises in a sequential, one-shot setup with Bayesian or non-Bayesian inferences, where decision orderings need to be carefully decided with respect to the network topology. 

Lastly, we remark that computing and approximating posterior probabilities in general Bayesian networks is known to be hard~\cite{Cooper1990-fn,Dagum1993-gd,Kwisthout2018-az}.
Our setting considers a restricted variant of a Bayesian network---for example, all private signals have the same probability of error.
Furthermore, we are interested in the average accuracy of the agents in the network, not their individual posterior probabilities.
Thus, to the best of our knowledge, there is no straightforward way to translate the existing complexity results for a general Bayesian network to our setting.

\section{Problem Statement}
\label{sec:problem}

This work studies two popular models of opinion exchange on networks. The overarching goal is for a network of truthful, rational agents to learn a binary piece of information, which we call the state of the world or \emph{ground truth}. We can also think of this state as an optimal binary action (buying or selling a stock, voting for a political party's candidate, etc.). The agents are arranged on a directed graph $G=(V,E)$ and broadcast which of the two states they believe is more probable to their neighbors. 

More explicitly, we encode the ground truth in $ \theta \in \{0,1\} $, distributed according to $ \berd{q} $, Bernoulli distribution with probability of $q$ of taking $1$ and $1-q$ of taking $0$. Every agent $v \in V$ initially receives an independent \emph{private signal} $ s_v \in \{0,1\}$ correlated with the ground truth. They then announce a prediction $a_v \in \{ 0,1 \} $ of the ground truth along outgoing edges, so out-neighbors of $v$ may use $a_v$ to improve their own predictions. Importantly, the probabilities $p$ and $q$, as well as the graph $G$ are all common knowledge. 
This is captured in the following formal definition of a network.

\begin{definition}[Social Network]
    A \emph{social network} is $ \network \deq ( G,q,p ) $,
    where \begin{enumerate}
        \item $ G = ( V,E ) $ is a directed graph with agents as vertices,
        \item $ q \in ( 0,1 ) $ is the prior probability of $\theta = 1$,
        \item $ p \in ( \frac 12, 1 ) $ is the accuracy of agents' private signals $s_v \in \{0,1\}$, such that \[\pr{s_v = 1 \suchthat \theta = 1} = \pr{s_v = 0 \suchthat \theta = 0} = p, \quad \forall v \in V.
            \]
    \end{enumerate}
    We further denote $ n \deq \absolute{V} $.
\end{definition}

We consider a classic asynchronous \emph{sequential} model ~\cite{Golub2017-qo}, in which agents announce their predictions in a \emph{decision ordering}, given by a one-to-one mapping $\sigma: V \to [n]$.
We denote the set of all possible orderings by $\Sigma_n$.
At every time step $i$, agent $v = \sigma^{-1}(i)$ makes an announcement $a_v \in \{0,1\}$.
The announcement depends on the agent's private measurement $ s_v $, along with the \emph{previous announcements of in-neighbors}, which we denote as a tuple $ N_v $, defined as \[
    N_v \deq ( a_u \suchthat u \in V \land uv \in E \land \sigma(u)<\sigma(v) ).
\] 
We call the tuple $X_v = (s_v) \cup N_v$ the \emph{inputs} of node $v$. 
This setup of limiting visibility to an agent's neighborhood has been studied in a number of recent papers~\cite{Bahar2020-am,arieli2020social,lu24enabling}.

When making announcements, agents follow an \emph{aggregation rule}, which is a function $ \mu:  ( X_v, G, \sigma ) \mapsto a_v $.
Broadly speaking, aggregation rules can either be Bayesian or non-Bayesian. 
In the \emph{Bayesian} model,  agents are fully rational
and make predictions according to their posterior probability for $\theta$, given their inputs and knowledge of the network topology $G$. In particular, agents take into account the correlation between their inputs resulting from the network topology and from the current ordering $ \sigma $. \[
     \mu^B(X_v, G, \sigma) \deq \begin{cases}
         1 & \text{if $\Pr_{G,\sigma}[\theta = 1 \suchthat X_v] > \frac 12$,} \\
         0 & \text{if $\Pr_{G,\sigma}[\theta = 0 \suchthat X_v] > \frac 12$,} \\
         \berd{\frac 12} & \text{otherwise.}
     \end{cases}
 \]

We also consider a non-Bayesian model, in which agents have bounded rationality and instead use simpler heuristic rules.
This is perhaps a more practical model, as computing posterior probabilities in arbitrary networks can become computationally expensive. 
In particular, we examine the \emph{majority dynamics} model, in which agents simply follow the majority among their inputs~\cite{Bahar2020-am,Shoham1992-ir,Laland2004-ej}. Since this model does not require agents to take into account correlations between their inputs derived from the network topology or the ordering, we omit $G$ and $ \sigma $ as inputs to $\mu^M$: \[
    \mu^M(X_v) \deq \begin{cases}
        1 & \text{if $  \frac 1{\absolute{X_v}}\sum_{x \in X_v} x > \frac 12 $,} \\
        0 & \text{if $  \frac 1{\absolute{X_v}}\sum_{x \in X_v} x < \frac 12 $,} \\
        s_v & \text{otherwise.}
    \end{cases}
\]

Finally, we quantify how successful the network is in predicting the ground truth by defining the following notion of a learning rate.
\begin{definition}
    The \emph{cumulative learning rate} (CLR) of a network $ \network $ under the ordering $ \sigma $ and an aggregation rule $ \mu $ is \[
		\clr (\network, \sigma, \mu) \deq \E_{\theta, s} \left[ \sum_{v \in V}^{} \mathds{1}_{\{a_v = \theta\}} \right] = \sum_{v \in V} \Pr_{\theta, s}\left[a_v = \theta\right],
	\]
    where the equality follows from linearity of expectation.
    Further, the \emph{learning rate} (LR) of a network $ \network $ under the ordering $ \sigma $ is simply \[
		\lr (\network, \sigma, \mu) \deq \tfrac 1 n \clr (\network, \sigma, \mu).
	\]
\end{definition}

We are mainly interested in the \emph{optimal} learning rate of a network, defined as follows.

\begin{definition}[Optimal LRs]
    The \emph{optimal cumulative learning rate} of a network $ \network $ is \[
        \oclr (\network, \mu) \deq \max_{\sigma \in \Sigma_n} \clr (\network, \sigma, \mu),
    \]
    and the \emph{optimal learning rate} of a network $ \network $ is \[
        \olr (\network, \mu) \deq \max_{\sigma \in \Sigma_n} \lr (\network, \sigma, \mu).
    \]
\end{definition}

Note that when $ p $ and $ q $ are clear from the context, we use the learning rate notation with only the graph, for example $ \olr (G, \mu) = \olr(( G,p,q ), \mu) $.
We can now present a formal definition of our main focus, the \netlearnopt{} optimization problem, and \netlearn{}, its decision version.

\begin{definition}[\netlearnopt{}]
    Suppose $\mu$ is a fixed aggregation rule.
    Given a network $ \network $, the \netlearnopt{} problem is to maximize  $ \lr (\network, \sigma, \mu) $, over $ \sigma \in \Sigma_n $.
\end{definition}

\begin{definition}[\netlearn{}]\label{def:decnetlearn}
    Suppose $\mu$ is a fixed aggregation rule.
    Given a network $ \network $ and a constant threshold $ \varepsilon \in (0,1)$,
    the \netlearn{} decision problem asks whether \[
            ( \exists \sigma \in \Sigma_n )\quad \lr (\network, \sigma, \mu) \geq 1-\varepsilon.
        \]
\end{definition}

Note that \Cref{def:decnetlearn} can be formulated equivalently by asking whether an optimal ordering $\sigma^*$ which maximizes the network learning rate achieves LR at least $1-\varepsilon$.

In \Cref{sec:bayes,sec:maj}, we focus on the decision problem, offering a proof that it is \np-hard for $ \mu = \mu^B $ and $ \mu = \mu^M $.
Finally, in \Cref{sec:approx}, we use insights from the \np-hardness proofs to show \netlearnopt{} is hard to even approximate.
Surprisingly, this gives us a stronger \np-hardness statement, showing that \netlearn{} is \np-hard even if we arbitrarily fix the agents' accuracy $ p \in ( \frac 12, 1 ) $.

\section{Proof of \np-hardness for the Bayesian Model}
\label{sec:bayes}

We now state one of the main results of this paper---the hardness of \netlearn{}.

\begin{theorem}[ ]\label{thm:nphardness_bayes}
	\netlearn{} with the Bayesian learning rule $\mu = \mu^B$ is \np-hard.
\end{theorem}

\subsection{Proof Idea}

We perform a reduction from \sat{} to \netlearn{}.
Specifically, we assume in our reduction that all formulas have exactly 3 distinct literals in each clause, and never a literal along with its negation.
For a given formula $ \varphi $, we construct a network $ \network $ and an $ \varepsilon > 0 $, such that
\[
    \olr(\network, \mu^B) \geq 1-\varepsilon \quad \iff \quad \text{$ \varphi $ is satisfiable.}
\]

The network $ \network $ consists of a directed graph $ G $, the ground truth prior $ q $, and the prior of the agents' private signals $ p $.
Our reduction requires setting $ p $ and $ G $ based on the formula, but it allows us to set $ q=\frac 12 $, regardless of $ \varphi $.

\subsection{Graph Construction \& Notation}
\label{ssec:graph}

First, we construct the directed graph $G$.
Our construction consists of $ N $ \emph{variable cells}, and $ M $ \emph{clause gadgets}, where $ N $ and $ M $ are the number of variables and clauses in $ \varphi $, respectively.
We define the cell and gadget first, and then define the full graph in \Cref{def:bayesian_graph}.

\begin{restatable}[Variable cell]{definition}{variableCell}
\label{def:cell}
    Let $ x $ be a variable of a formula $ \varphi $.
    A \emph{variable cell} of $ x $ is a directed graph $ \cell x = \left( V_x, E_x \right) $, where \begin{enumerate}[ ]
    	\item $ V_x = \left\{ x, \lnot x, d_x \right\} $, and
	\item $ E_x = \left\{ \left( d_x, x \right), \left( d_x, \lnot x \right), \left( x, \lnot x \right), \left( \lnot x, x \right) \right\} $.
    \end{enumerate}
\end{restatable}

\Cref{fig:cell_bayesian} depicts a cell for some variable $x_i$.

\begin{definition}[Clause gadget]
    Let $ C = j \lor k \lor \ell $ be a clause of a 3-CNF formula $ \varphi $, where $ j \neq k \neq \ell $ are some literals.
    Then the \emph{clause gadget} is $ \gadget C = \left( V_C, E_C \right) $, where \begin{enumerate}[ ]
    	\item $ V_C = \left\{ j,k,\ell \right\} $, and
    	\item $ E_C = \left\{ \left( x,y \right) \suchthat x,y \in \left\{ j,k,\ell \right\} \land x \neq y \right\} $.
    \end{enumerate}
\end{definition}

\begin{restatable}[Formula graph]{definition}{bayesianGraph}\label{def:bayesian_graph}
	Let $ \varphi = C_1 \land C_2 \land \dots \land C_M $ be a CNF formula of variables $ \vars = \left\{ x_1, \ldots, x_N \right\} $, where each clause $ C_i $ is the disjunction of exactly three literals.
	Then $ G_\varphi $ is a disjoint union of the graphs $ \cell {x_i} $ for all $ {x_i} \in \vars $, and $ \gadget {C_i} $ for all clauses $ C_i \in \varphi $.
        Additionally, there is an edge to each of the vertices in every $ \gadget {C_i} $ from the corresponding literal nodes in the respective variable cell.
\end{restatable}

Note that that there are no incoming edges to any cell, except from within the same cell, so the learning rate of each cell is determined only by the ordering of its own vertices. Also, no two clause gadgets share any vertices, so the ordering of vertices within a gadget only affects that gadget.
For an illustration of the formula graph construction, refer to \Cref{fig:gphi_bayesian}, where we give a sample formula graph for the formula $ \varphi = x \lor y \lor \lnot z $.

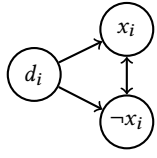
\begin{figure}[t!]
	\centering
	\begin{tikzpicture}[
		xsh/.style = { xshift=10mm },
		node distance = 5mm and 5mm,
	stff/.style={circle, draw=black, \figThickness, minimum size=\figCircSize, inner sep=0pt},
	]
		\node[stff]        (x)                  {$ x_i $};
		\node[stff]        (nx)   [below=of x]   {$ \lnot x_i $};
		\node[stff]        (d)   [left=of x,yshift=-6mm]   {$ d_i $};

		\draw[->, \figThickness] (d)  to  (x);
		\draw[->, \figThickness] (d)  to (nx);
		\draw[<->, \figThickness] (x)  to (nx);
	\end{tikzpicture}
 \Description[A graph of three vertices, connected together]{A graph containing vertices x and non-x, connected both-ways, along with a degree d, from which two edges go to x and non-x, but not the other way around.}
	\caption{The cell for variable $ x_i $.}
 \label{fig:cell_bayesian}
\end{figure}
\begin{figure}[t!]
	\centering
	\begin{tikzpicture}[
		xsh/.style = { yshift=-10mm },
		node distance = 5mm and 5mm,
		dcs/.style = { yshift=-3mm,xshift=-6mm },
	stff/.style={circle, draw=black, \figThickness, minimum size=\figCircSize, inner sep=0pt},
	]
		\node[stff]        (x_i)                  {$ x $};
		\node[stff]        (nx_i)   [left=of x_i]   {$ \lnot x $};
		\node[stff]        (d_i)   [above=of x_i,dcs]   {$ d_x $};

		\node[stff]        (x_j)    [left=of nx_i]  {$ y $};
		\node[stff]        (nx_j)   [left=of x_j]   {$ \lnot y $};
		\node[stff]        (d_j)   [above=of x_j,dcs]   {$ d_y $};

		\node[stff]        (x_k)    [left=of nx_j]  {$ z $};
		\node[stff]        (nx_k)   [left=of x_k]   {$ \lnot z $};
		\node[stff]        (d_k)   [above=of x_k,dcs]   {$ d_z $};

		\node[stff]        (l_m^1)    [below=of x_i,xsh]  {$ x_C $};
		\node[stff]        (l_m^2)    [below=of x_j]  {$ y_C $};
		\node[stff]        (l_m^3)    [below=of nx_k,xsh]  {$ \lnot z_C $};

		\draw[->, \figThickness] (d_i)  to  (x_i);
		\draw[->, \figThickness] (d_i)  to (nx_i);
		\draw[<->, \figThickness] (x_i)  to (nx_i);

		\draw[->, \figThickness] (d_j)  to  (x_j);
		\draw[->, \figThickness] (d_j)  to (nx_j);
		\draw[<->, \figThickness] (x_j)  to (nx_j);

		\draw[->, \figThickness] (d_k)  to  (x_k);
		\draw[->, \figThickness] (d_k)  to (nx_k);
		\draw[<->, \figThickness] (x_k)  to (nx_k);

		\draw[->, \figThickness] (x_i)  to (l_m^1);
		\draw[->, \figThickness] (x_j)  to (l_m^2);
		\draw[->, \figThickness] (nx_k)  to (l_m^3);

		\draw[<->, \figThickness] (l_m^1)  to (l_m^2);
		\draw[<->, \figThickness] (l_m^1)  to (l_m^3);
		\draw[<->, \figThickness] (l_m^2)  to (l_m^3);
	\end{tikzpicture}
	\caption{The graph $ G_\varphi $ for $ \varphi = C = \left( x \lor y \lor \lnot z \right)$.}
 \Description{A construction for the Majority vote reduction, as described in \Cref{def:bayesian_graph} for a simple formula.}
	\label{fig:gphi_bayesian}
\end{figure}
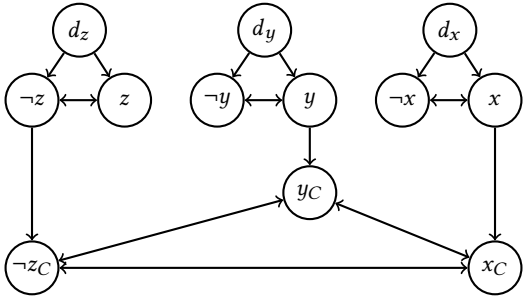

\ourparagraph{Ordering-assignment relation}\label{par:ordering_literals}
To determine satisfiability of $\varphi$ from the learning rate of $ G_\varphi $, we map vertex orderings to variable assignments.
We then show that orderings achieving higher learning rates correspond to assignments with more satisfied clauses.
For a more detailed description of both the mapping and its properties, see \Cref{ssec:assignment_bayes}.

For clarity, we introduce some more notation. Let $ \ell $ be a literal of some variable $ x \in \vars $, meaning $ \ell = x $ or $ \ell = \lnot x $.
We say that cell $\cell x$ is in one of two states: it is ``on'' under a decision ordering $\sigma$ if $\sigma(\lnot x) < \sigma(x)$; otherwise, cell $\cell x$ is ``off''.
We say that the literal $ \ell $ is ``on'' if $ \ell = x $ and $ \cell x $ is on, or $ \ell = \lnot x $ and $ \cell x $ is off; otherwise, literal $ \ell $ is ``off''.
For brevity, we further denote $ \cell \ell \deq \cell x $, regardless of whether $ \ell = x $ or $ \ell = \lnot x $.

\subsection{Gadget Learning Rates}
This section lists the learning rates of the cells and clause gadgets under Bayesian aggregation. We assume WLOG for this section that $\theta = 1$, and compute all probabilities in this section conditioned on $\theta = 1$. Since we always take $\theta$ to be uniform on $\{0,1\}$, this yields the same values as taking the probability over $\theta$ as well. We note that our computations were verified using Wolfram Mathematica.

We begin by examining the learning rate of an arbitrary cell under a pair of orderings in which the cell is either ``on'' or ``off''. The following lemma shows that cells achieve the same learning rate under either of these orderings, and that this learning rate is the best possible over all orderings. 

\begin{lemma}[Bayesian Cell LR] \label{lemma:bayesian_cellLearningRate}
    Let $x \in \vars$.
    Let $ q = \frac 12 $, and $ p > \frac 12 $ be given.
    Then \[
	\oclr (\cell x) = \tfrac 52 p + \tfrac 32 p^2 - p^3.
    \]
    In particular, if under an optimal ordering $ \sigma^* $ a literal $ \ell $ is on, then its corresponding literal node has learning rate $ \clr (\ell, \sigma^*, \mu^B) = \frac p2 + \frac 32 p^2 - p^3 $; otherwise $ \clr (\ell, \sigma^*, \mu^B) = p $.
\end{lemma}

\begin{proof}
    First, observe that for the optimal ordering $ \sigma^* $, it is always beneficial to put the dummy node $ d_x $ \emph{before} the nodes $ x $ and $ \lnot x $.
    This is because $ d_x $ has no incoming edges, so it cannot acquire more information by going later, and it has edges going to $ x $ and $ \lnot x $, which can only increase their chances of getting the correct answer.
    Hence, we can see that $ \oclr(d_x) = p $.

    The case of the remaining two nodes is symmetric, so WLOG, let us assume that $ \sigma^*(x) < \sigma^*(\lnot x) $ (so the cell $ \cell x $ is ``off'').
    The node $ x $ then receives the action of $ d_x $, which is i.i.d. from its private information.
    So node $ x $ chooses its action correctly either if both $ s_x $ and $ a_{d_x} $ are correct, where $a_{d_x}$ is the action chosen by node $d_x$, or if exactly one of the two are correct and $ x $ tiebreaks correctly. The first outcome occurs with probability $ p^2 $, and the second with probability $ 2p (1-p) \frac 12 $.
    Thus, the learning rate of node $ x $ is \[
        \oclr(x) = p^2 + p - p^2 = p.
    \]

    Finally, the node $ \lnot x $ receives its private signal, $ s_{\lnot x} $ and the actions $ a_{d_x} $ and $ a_x $.
    Notice that these three pieces of information are \emph{not} independent, since $ x $ was influenced by $ d_x $.
    We perform case analysis on the relative likelihood \[
        \Lambda \deq \frac{\pr{\theta = 1 \mid X_{\lnot x}}}{\pr{\theta = 0 \mid X_{\lnot x}}} = \frac{\pr{X_{\lnot x} \mid \theta = 1}}{\pr{X_{\lnot x} \mid \theta = 0}},
    \]
    where the equality follows from Bayes' theorem and from the fact that the prior $ q = \frac 12 $.
    Recall that $ X_{\lnot x} = \{s_{\lnot x}, a_x, a_{d_x}\} $, and that $ s_{\lnot x} $ is independent of the other two.
    Thus, $ \Lambda $ can be expressed as
    \begin{align*}
        \Lambda = \frac{\pr{s_{\lnot x} \mid \theta = 1}}{\pr{s_{\lnot x} \mid \theta = 0}} \cdot
        \frac{\pr{a_{d_x} \mid \theta = 1}}{\pr{a_{d_x} \mid \theta = 0}} \cdot
        \frac{\pr{a_x \mid \theta = 1, a_{d_x}}}{\pr{a_x \mid \theta = 0, a_{d_x}}}.
    \end{align*}
    We compute $ \Lambda $ for each case of $(s_{\lnot x}, a_x, a_{d_x})$:
    \begin{enumerate}
        \item $ \left( 1,1,1 \right) $: \[
                \Lambda = \tfrac p{1-p} \cdot \tfrac p{1-p} \cdot \tfrac{p + (1-p)/2}{(1-p) + p/2} = \tfrac {p^2}{(1-p)^2} \cdot \tfrac {1/2+ p/2}{1- p/2} > 1.
        \]
        \item $ \left( 1,1,0 \right) $: \[
                \Lambda = \tfrac p{1-p} \cdot \tfrac {1-p}p \cdot \tfrac{p/2}{(1-p)/2} = \tfrac {p}{(1-p)} > 1.
        \]

        \item $ \left( 1,0,1 \right) $: \[
                \Lambda = \tfrac p{1-p} \cdot \tfrac p{1-p} \cdot \tfrac{(1-p) /2}{p/2} = \tfrac {p}{(1-p)} > 1.
        \]

        \item $ \left( 1,0,0 \right) $: \[
                \Lambda = \tfrac p{1-p} \cdot \tfrac {1-p}p \cdot \tfrac{(1-p) + p/2}{p + (1-p) /2} =  \tfrac {1-p/2}{1/2+ p/2} < 1.
        \]
    \end{enumerate}
    The inequalities hold for $ \tfrac 12 < p < 1 $.
    The remaining cases (that is $(0,1,1)$, $(0,1,0)$, $(0,0,1)$, $(0,0,0)$) follow from symmetry.
    It is now clear that if $ \theta = 1 $, $ a_{\lnot x} $ is correct in cases 1, 2, 3, 5.
    We now compute $ \oclr(\lnot x) $, which is the probability of cases 1, 2, 3, or 5 occurring. \begin{align*}
        \oclr(\lnot x) &= \pr{X_{\lnot x} \in \left\{ \left( 1,1,1 \right) ,  \left( 1,0,1 \right) ,  \left( 1,1,0 \right) ,  \left( 0,1,1 \right) \right\}} \\
                       &= p ( p + (1-p)\tfrac p2)) + (1-p)p(p + (1-p)\tfrac 12) \\
                       &= \tfrac p2 + \tfrac 32 p^2 - p^3.
    \end{align*}

    The cumulative learning rate of the clause gadget is then the sum of the above, \[
        \oclr(\cell x) = \tfrac 52 p + \tfrac 32 p^2 - p^3.
    \]
\end{proof}

To summarize, there exist two orderings for each cell which yield the same cell learning rate.
The two orderings place $d_x$ first, and correspond to either the ``on'' state (when $ \sigma(x) < \sigma(\lnot x) $) or the ``off'' state ($ \sigma(x) > \sigma(\lnot x) $).
For now, we will defer the question of which of the two is optimal in the overall network ordering.
Addressing this first requires examining the learning rates of the clause gadgets.
Since the graph contains edges from cells to clause gadgets, there is an optimal ordering which orders all cell nodes before all clauses gadgets.
We begin by examining the learning rate for arbitrary clause gadgets. 

\begin{lemma}\label{lemma:bayesian_clause}
    Let $ C = \alpha \lor \beta \lor \gamma $ be a clause of $ \varphi $.
    Let $ \sigma^* $ be an optimal ordering on $ \network $.
    Let $p'\deq\frac p2 + \frac 32 p^2 - p^3 $.
    Then if under $ \sigma^* $, exactly $ i $ cells of $ \alpha, \beta, \gamma $ are ``on'', then $ \clr(\gadget C, \sigma^*) = \clr_i $, where \begin{align*}
        \pzero &\deq p (2 p^4-5 p^3+5 p+1). \\
        \pone &\deq  p^4 (2 p'-1)+p^3 (2-4 p')+p^2 (1-2 p')+4 p p'+p', \\
        \ptwo &\deq 4 p^3 (p'-1) p'+p^2 (-6 (p')^2+4 p'+1)+2 p p'+p' (p'+1), \\
        \pthree &\deq p' (p^2 (2 (p')^2-3 p'+1)-p (2 (p')^2+p'-3)+2 p'+1).
    \end{align*}
    Furthermore, if $ D \in \varphi $ is a satisfied clause under the ordering $ \sigma^* $, meaning at least one of the literals of $ D $ is ``on'' under $ \sigma^* $, then 
    \[
        \pthree \geq \clr(\gadget D, \sigma^*) \geq \pone \geq \pzero.
    \]
\end{lemma}

\begin{figure}[t!]
	\centering
\begin{tikzpicture}
	\begin{axis}[
    xlabel={$ p $},
    ylabel={$ \clr(\gadget C) $},
    legend pos=south east,
    axis lines=left,
]

\addplot [
    domain=1/2:1, 
    samples=\figSamples,
    color=red,
    smooth
]
{x*(2+14*x+27*x^2-6*x^3-59*x^4-7*x^5+62*x^6+21*x^7-78*x^8+44*x^9-8*x^10)/4};
\addlegendentry{\(\pthree\)}

\addplot [
    domain=1/2:1, 
    samples=\figSamples,
    color=green,
    smooth
]
{x*(2+15*x+22*x^2+7*x^3-84*x^4+14*x^5+92*x^6-72*x^7+16*x^8)/4};
\addlegendentry{\(\ptwo\)}

\addplot [
    domain=1/2:1, 
    samples=\figSamples,
    color=blue,
    smooth
]
{(x+9*x^2+12*x^3-20*x^4-6*x^5+14*x^6-4*x^7)/2};
\addlegendentry{\(\pone\)}

\addplot [
    domain=1/2:1, 
    samples=\figSamples,
    color=black,
    smooth
]
{x*(1+5*x-5*x^3+2*x^4)};
\addlegendentry{\(\pzero\)}

\end{axis}
\end{tikzpicture}
\caption{The relationship between learning rates of $ \gadget C $, depending on the number of ``on'' literals, as a function of $ p $.}
\Description{All are equal to $\frac 12$ for $p=\frac 12$, and $1$ for $p=1$. In the interval between those two values, $\pzero \leq \pone \leq \ptwo \leq \pthree$.}
\label{fig:bayes_ps}
\end{figure}
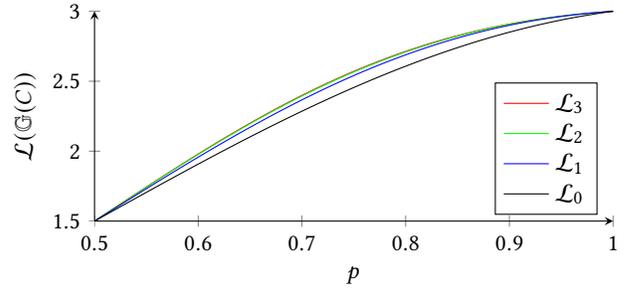

\begin{proof}[Proof Idea.]
    The full proof is long and technical.
    Here we offer the main idea, which is similar to that of \Cref{lemma:bayesian_cellLearningRate}.

    The actions of the nodes in $ \gadget C $ are exactly determined by \begin{enumerate}[ ]
        \item the ordering of the nodes in $ \gadget C $,
        \item the private signals of the nodes in $ \gadget C $,
        \item the actions taken by $ \alpha, \beta, \gamma $, as well as $ \alpha, \beta, \gamma $ being on or off.
    \end{enumerate}
    We compute the learning rate in each case.
    We then compute the expected value over the private signals and the actions of $ \alpha,\beta,\gamma $, the probabilities of which are given by \Cref{lemma:bayesian_cellLearningRate}.

    The resulting cumulative learning rate of $ \gadget C $ depends only on the states of $ \alpha, \beta, \gamma $.
    Among these are the learning rates $ \pone $ and $ \pthree $, which are the learning rates of the clause gadget when one or all of the literals are on, respectively.
     
    The case of $\pone$ actually corresponds to three sub-cases, depending on whether the node corresponding to the ``on'' literal is first, second, or third in the ordering. Notice that swapping the ordering of the three vertices inside the clause gadget does not affect the learning rate of other vertices in the network. Thus, since $ \sigma^* $ is an optimal ordering, it maximizes the learning rate for this clause, and as such, $\pone$ is the maximum over the three subcases.

    The full proof can be found in \appbayes.
\end{proof}

See \Cref{fig:bayes_ps} for the relationship between the different values as a function of $ p $.
Notice that $ \pthree,\ptwo,\pone$ and $\pzero $ converge as $p$ approaches $1$ and as it approaches $\frac 12$, regardless of the graph topology and ordering.
This is exactly what we expect: if $ p=1 $, the agents have perfect information from their private signals alone, while if $ p=\frac 12 $, then the private signals give the agents no extra information, and thus their the LR approaches $ \frac 12 $.

\subsection{Optimal Ordering \& Restrictions on p}
\label{ssec:assignment_bayes}
We can now determine the optimal ordering by examining the learning rates derived in the previous section. First, we define an \emph{induced ordering} below:

\begin{definition}
    Let $\mathcal{A}: \chi \to \{0,1\}$ be any assignment of values to variables. Define the (partial) \emph{ordering $\sigma(\mathcal{A})$ induced by $\mathcal{A}$} as follows: if $\mathcal{A}(x_i)=1$, then cell $\cell {x_i}$ is on; otherwise, $\cell {x_i}$ is off.
\end{definition}

In particular, the above definition gives a bijection between assignments and partial ordering over variables. We further say that a total ordering $\sigma$ \emph{respects} an assignment $\mathcal{A}$ (denoted by $\sigma \sim \mathcal{A}$) if it contains $\sigma(\mathcal{A})$ as a partial ordering over variable nodes. We also write $\sigma^*(\mathcal{A}) \deq \arg\max_{\sigma \sim \mathcal{A}} \mathcal{L}(\mathcal{N}, \sigma, \mu^B)$.

\begin{definition} \label{def:max_assignment}
    Let $\mathcal{A}: \chi \to \{0,1\}$ be an assignment of values to variables maximizing the number of satisfied clauses. Then $\mathcal{A}$ is a \emph{maximal assignment}.
\end{definition}

\begin{lemma}[Optimal Ordering] \label{lemma:bayes_bestOrder}
    Let $\mathcal{A}^*$ be a maximal assignment.  Let $p(M) \deq (3M-4)/(3M-3)$ be a threshold probability determined by $M$, the number of clauses. Then for all $p \geq p(M) $, $\sigma^*(\mathcal{A}^*)$ is an optimal ordering.
\end{lemma}
\begin{proof}
    Note that an assignment-induced ordering only specifies whether each cell is on or off. From \Cref{lemma:bayesian_cellLearningRate},  an optimal cell ordering exists both if the cell is on or off, and yields the same learning rate in either case. Therefore, the optimal ordering is determined by comparing clause learning rates. 
    
    Also note that any total ordering must respect some assignment. So we argue the optimality of $\sigma^*(\mathcal{A}^*)$ by comparing it to $\sigma^*(\mathcal{A}')$ for any non-maximal $\mathcal{A'}$. Let $S^*$ be the number of satisfied clauses under $\mathcal{A}^*$ and $S' < S^*$ that under $\mathcal{A'}$. By \Cref{lemma:bayesian_clause}, for any clause satisfied under an arbitrary assignment $\mathcal{A}$, the corresponding clause gadget under an ordering $\sigma \sim \mathcal{A}$ achieves learning rate lower bounded by $\pone$. Further, by \Cref{lemma:bayesian_clause}, $\pone >\pzero$,  so having more satisfied clauses in an assignment can never decrease the CLR under the induced ordering. However, note also that $\pthree >\pone$  (by \Cref{lemma:bayesian_clause}). Hence, in the most extreme case, all $S'$ satisfied clauses under $\mathcal{A}'$ are satisfied with three true literals, while all $S^*$ satisfied clauses under $\mathcal{A}^*$ are satisfied with one true literal. 

    Consider that extreme case, and further impose the worst-case choices of $S'$ and $S^*$ by setting $S^* = M$ and $S' = M-1$. Then over all clause gadgets, $\sigma^*(\mathcal{A}^*)$ gives a CLR of $M \pone$, $\sigma^*(\mathcal{A}')$ gives a CLR of $(M-1) \pthree + \pzero$. We can now solve for conditions on $p$ such that $\sigma^*(\mathcal{A}^*)$ achieves a higher network learning rate:
    \begin{align*}
        M \pone &> (M-1) \pthree + \pzero\\
        \tfrac{\pone - \pzero}{\pthree - \pone} &> M-1.
    \end{align*}
    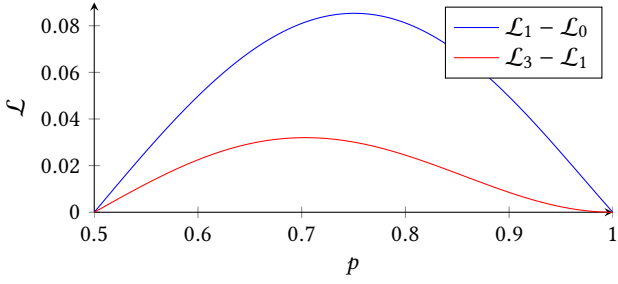
\begin{figure}[t!]
	\centering
\begin{tikzpicture}
	\begin{axis}[
    xlabel={$ p $},
    ylabel={$ \clr $},
    axis lines=left,
    ymax=0.09,
    scaled ticks=false,
    yticklabel style={
	/pgf/number format/precision=3,
	/pgf/number format/fixed,
    }
]

\addplot [
    domain=1/2:1, 
    samples=\figSamples,
    color=blue,
    smooth
]
{-(1/2)*x*(1+x-12*x^2+10*x^3+10*x^4-14*x^5+4*x^6)};
\addlegendentry{\(\pone-\pzero\)}

\addplot [
    domain=1/2:1, 
    samples=\figSamples,
    color=red,
    smooth
]
{-(1/4)*(-1+x)^2*x^2*(4+5*x-28*x^2-14*x^3+35*x^4+14*x^5-28*x^6+8*x^7)};
\addlegendentry{\(\pthree-\pone\)}

\end{axis}
\end{tikzpicture}
\caption{The plot of the numerator and the denominator of the condition on $ M-1 $ from \Cref{lemma:bayesian_clause}.
The denominator approaches zero noticably faster, suggesting the fraction approaches infinity as $ p $ goes to 1.}
\label{fig:pdff}
\Description{The difference $\pone-\pzero$ is always strictly larger than $\pthree-\pone$. As $p$ goes to 1, their ratio grows.}
\end{figure}

    The numerator and denominator is plotted in \Cref{fig:pdff} as a function of $ p $.
    Since the left-hand side can be lower bounded by $\frac{1}{3-3p}$ for all $p \in (0.5,1)$, it is sufficient to have $ \frac 1{3-3p} \geq M-1 $.
    We can thus define $ p(M) \deq \frac{3M-4}{3M-3} $, which respects the condition. Note that $p(M)$ is well-defined for any $ M \geq 2 $, and lies in the interval $ (\frac 12, 1) $.
\end{proof}

To recap, given any instance $\varphi$ of \sat{} with $N$ variables and $M$ clauses, we can construct a network $G_\varphi$ with ground truth prior $q = 1/2$ and signal accuracy $p = p(M)$. In particular, whenever $\varphi$ is satisfiable under some maximal assignment $\mathcal{A}^*$, there is an optimal decision ordering $\sigma^*$ which respects the ordering induced by $\mathcal{A}^*$, and which achieves a network CLR of at least $N(2p+3p^2-2p^3) + M\pone$. Otherwise, if $\varphi$ is non-satisfiable under any maximal assignment $\mathcal{A}^*$, then the optimal decision ordering $\sigma^*$ respecting the ordering induced by $\mathcal{A}^*$ achieves a network CLR of no more than $N(2p+3p^2-2p^3) + (M-1)\pthree + \pzero$. Choosing $p = p(M)$ allowed us to show the optimality of $\sigma^*$, as well as to separate the learning rates in networks corresponding to satisfiable and non-satisfiable formulas. All that remains to complete the reduction is to pick an appropriate choice of $\varepsilon$, such that a formula graph $G_\varphi$ achieves expected network learning rate greater than $\varepsilon$ under an optimal ordering iff $\varphi$ is satisfiable.

\subsection{Picking the Epsilon}\label{sec:bayes_epsilon}
We will simply pick $\varepsilon$ to lie exactly halfway between the satisfiable and non-satisfiable network learning rates. Recalling that each variable and clause gadget contains $3$ vertices, we can compute the learning rates below. For networks corresponding to satisfiable formulas, we have $\frac{N\pcell + M\pone}{3(N+M)},$ and for networks corresponding to non-satisfiable formulas, we have $\frac{N\pcell + (M-1)\pthree + \pzero}{3(N+M)}.$ This gives 
\begin{align*}
    \varepsilon = \frac{1}{2}\left(\frac{2N\pcell + M\pone + (M-1)\pthree + \pzero}{3(N+M)}\right),
\end{align*}
thus completing the reduction. 

\subsection{Generalization for a constant range of p}

Finally, we remark that, while this construction proves the Theorem, it uses $ p $ approaching 1 as the size of $ \varphi $ increases.
A natural question, then, is whether it is hard to compute the maximum learning rate for other values of $p$, especially when $ p $ is a fixed constant.
In fact, \Cref{sec:approx} shows a stronger version of \Cref{thm:nphardness_bayes}, stated as follows. 

\begin{theorem}[ ]
	\netlearn{} with Bayesian learning rule $\mu = \mu^B$ and a fixed $ p \in ( \sqrt{7/8},1 ) $ is \np-hard.
\end{theorem}

This follows as a direct corollary of \Cref{thm:approx}.

\section{Proof of \np-hardness for Majority Dynamics}
\label{sec:maj}

In this section, we adapt the statement of \Cref{thm:nphardness_bayes} to the Majority dynamics setting.

\begin{theorem}[ ]\label{thm:nphardness_maj}
	\netlearn{} with the majority vote rule $\mu = \mu^M$ is \np-hard.
\end{theorem}

The proof is again a reduction from \sat.
The main idea is identical to that of \Cref{thm:nphardness_bayes}, so we offer here only the main points and construction, with emphasis on differences from the proof of \Cref{thm:nphardness_bayes}.
We offer the full proof in \appmaj.

\subsection{Adapted Graph Construction}
\label{ssec:maj_graph}

Unfortunately, if we wanted to directly apply the previous construction, we would find that the worst-case satisfied ordering reaches a \emph{lower} learning rate than the best-case unsatisfied ordering (as discussed in \Cref{ssec:assignment_bayes}).
Then there is no $ \varepsilon $ which separates the learning rates corresponding to satisfied and unsatisfied assignments.
We thus adapt the construction from \Cref{ssec:graph}, modifying the clause gadget.
We keep the variable cell unchanged.

\variableCell*

Intuitively, we need to give the nodes in the clause gadgets more input, so they are better equipped to use the input from the ``on'' cells.
To achieve this, we add two dummy nodes to the clause gadget.

\begin{definition}[Clause gadget]
    Let $ C = j \lor k \lor \ell $ be a clause of a 3-CNF formula $ \varphi $, where $ j \neq k \neq \ell $ are some literals.
    Then the \emph{clause gadget} is $ \gadget C = \left( V_C, E_C \right) $, where \begin{enumerate}[ ]
    	\item $ V_C = \left\{ j,k,\ell, d^1, d^2 \right\} $,
    	\item $ E_C = \{ ( x,y ) \mid x,y \in \{ j,k,\ell \}, x \neq y \} \cup \{ ( d^i, x ) \mid i \in \{ 1,2 \} , x \in \{ j,k,\ell \} \} $.
    \end{enumerate}
\end{definition}

We now define the formula graph using the same definition, only with the new clause gadget.

\bayesianGraph*

See \Cref{fig:gphi} for an illustration of this construction for the simple formula $ \varphi = x \lor y \lor z $.
Note that we also use the ``on''/``off'' states of a variable/literal, as defined for the Bayesian proof (see \Cref{par:ordering_literals}).

\begin{figure}[t!]
	\centering
	\begin{tikzpicture}[
		xsh/.style = { yshift=-9mm },
		xxsh/.style = { yshift=-8mm },
		dcs/.style = { yshift=-3mm,xshift=-6mm },
		node distance = 5mm and 5mm,
	stff/.style={circle, draw=black, \figThickness, minimum size=\figCircSize, inner sep=0pt},
	]
		\node[stff]        (x_i)                  {$ x $};
		\node[stff]        (nx_i)   [left=of x_i]   {$ \lnot x $};
		\node[stff]        (d_i)   [above=of x_i,dcs]   {$ d_x $};

		\node[stff]        (x_j)    [left=of nx_i]  {$ y $};
		\node[stff]        (nx_j)   [left=of x_j]   {$ \lnot y $};
		\node[stff]        (d_j)   [above=of x_j,dcs]   {$ d_y $};

		\node[stff]        (x_k)    [left=of nx_j]  {$ z $};
		\node[stff]        (nx_k)   [left=of x_k]   {$ \lnot z $};
		\node[stff]        (d_k)   [above=of x_k,dcs]   {$ d_z $};

		\node[stff]        (l_m^1)    [below=of x_i,xsh]  {$ x_C $};
		\node[stff]        (l_m^2)    [below=of x_j]  {$ y_C $};
		\node[stff]        (l_m^3)    [below=of nx_k,xsh]  {$ \lnot z_C $};
		\node[stff]        (d1)    [below=of l_m^2, xxsh,xshift=20]  {$ d_C^1 $};
		\node[stff]        (d2)    [below=of l_m^2, xxsh,xshift=-40]  {$ d_C^2 $};

		\draw[->, \figThickness] (d_i)  to  (x_i);
		\draw[->, \figThickness] (d_i)  to (nx_i);
		\draw[<->, \figThickness] (x_i)  to (nx_i);

		\draw[->, \figThickness] (d_j)  to  (x_j);
		\draw[->, \figThickness] (d_j)  to (nx_j);
		\draw[<->, \figThickness] (x_j)  to (nx_j);

		\draw[->, \figThickness] (d_k)  to  (x_k);
		\draw[->, \figThickness] (d_k)  to (nx_k);
		\draw[<->, \figThickness] (x_k)  to (nx_k);

		\draw[->, \figThickness] (x_i)  to (l_m^1);
		\draw[->, \figThickness] (x_j)  to (l_m^2);
		\draw[->, \figThickness] (nx_k)  to (l_m^3);

		\draw[<-, \figThickness] (l_m^1)  to  (d1);
		\draw[<-, \figThickness] (l_m^2)  to  (d1);
		\draw[<-, \figThickness] (l_m^3)  to  (d1);
		\draw[<-, \figThickness] (l_m^1)  to  (d2);
		\draw[<-, \figThickness] (l_m^2)  to  (d2);
		\draw[<-, \figThickness] (l_m^3)  to  (d2);
		\draw[<->, \figThickness] (l_m^1)  to (l_m^2);
		\draw[<->, \figThickness] (l_m^1)  to (l_m^3);
		\draw[<->, \figThickness] (l_m^2)  to (l_m^3);
	\end{tikzpicture}
	\caption{The graph $ G_\varphi $ for $ \varphi = C = \left( x \lor y \lor \lnot z \right)$.}
 \Description{A construction for the Majority vote reduction, as described in \Cref{def:bayesian_graph} for a simple formula.}
	\label{fig:gphi}
\end{figure}
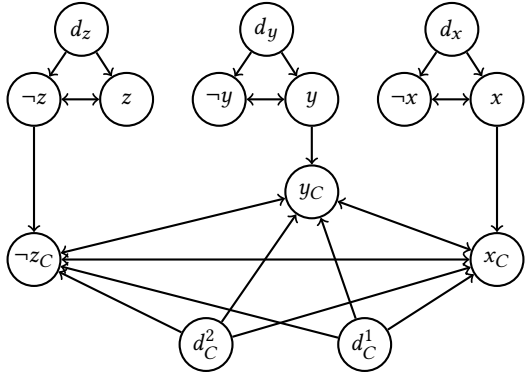

Next, we compute the learning rates of the variable cell, and of the clause gadgets, depending on whether its literals are on or off.

\begin{lemma}[Majority Dynamics Cell LR] \label{lemma:maj_MD_cellLearningRate}
    Let $x \in \vars$.
    Let $ q = \frac 12 $, and $ p $ be given.
    Then \[
	\oclr (\cell x) = 2p + 3p^2 -2p^3.
    \]
\end{lemma}

\begin{lemma}[Majority Dynamics Gadget Learning Rate]\label{lemma:maj_100CLR} 
    Let $C$ be a clause.
    Suppose that $\sigma^*$ is an optimal learning rate.
    Then, in the gadget for $C$, $\sigma^*$ places the cells first, then the dummy nodes, and finally the three literal nodes.
    Further, \begin{enumerate}
        \item if one literal is on, then $\clr(\gadget C, \sigma^*) $ is
        \begin{align*}
                \pone \deq p &\left(2 + 2 p + 6 p^2 + 11 p^3 + 4 p^4 - 51 p^5 - 6 p^6  + 21 p^7 \right. \\
                               & \left. + 115 p^8 - 136 p^9 + 13 p^{10} + 36 p^{11} - 12 p^{12} \right).
            \end{align*}
        \item if one literal is on, then $\clr(\gadget C, \sigma^*)$ is \begin{align*}
                \ptwo \deq p & \left(12 p^8-54 p^7+76 p^6-14 p^5 \right. \\
                               & \left. -40 p^4+9 p^3+12 p^2+2 p+2\right).
            \end{align*}
        \item if one literal is on, then $\clr(\gadget C, \sigma^*) $ is \begin{align*}
                \pthree \deq p & \left(2 + 2 p + 3 p^2 + 14 p^3 + 22 p^4 - 66 p^5 - 69 p^6 + 310 p^7 \right.\\
                            & \left.- 688 p^8 + 710 p^9 + 756 p^{10} - 2581 p^{11} + 2304 p^{12} - 558 p^{13}\right.\\ 
                            & \left.- 372 p^{14} + 264 p^{15} - 48 p^{16} \right).
            \end{align*}
    \end{enumerate}
    Furthermore, for a clause $ D $ satisfied under $ \sigma^* $, it holds \[
	    \pthree \geq \clr(\gadget D, \sigma^*) \geq \pone \geq \pzero.
	\]
\end{lemma}

Finally, we can now determine the optimal ordering, compute its learning rate, and show that there is an $ \varepsilon $ such that the learning rate is above $ \varepsilon $ if and only if the induced ordering is satisfied.

\begin{lemma}[Optimal Ordering] \label{lemma:maj_bestOrder}
    Let $\mathcal{A}^*$ be a maximal assignment.  Let $p(M) < 1$ be a threshold probability determined by $M$, the number of clauses. Then for all $p \geq p(M)$, the decision ordering $\sigma^*$ which places all dummy nodes first, then all variable nodes respecting the partial ordering induced by $\mathcal{A}^*$, and finally all literal nodes in the clause gadgets, maximizes the network learning rate.
\end{lemma}

The proof of this Lemma is very similar that of \Cref{lemma:bayes_bestOrder}, and is included in \appmaj.

We now apply the same reasoning as in \Cref{sec:bayes_epsilon}.
If $ p = p(M) $, then the worst-case learning rate of a satisfying assignment is $\frac{N\pcell + M\pone}{3N + 5M},$ strictly higher than the best-case non-satisfying LR, $\frac{N\pcell + (M-1)\pthree + \pzero}{3N + 5M}$. We can thus define the threshold (mind the new number of vertices---$ G_\varphi $ now has 5 for each clause):
\begin{align*}
    \varepsilon = \frac{1}{2}\left(\frac{2N\pcell + M\pone + (M-1)\pthree + \pzero}{3N + 5M}\right).
\end{align*}
This concludes the proof.

\section{Approximating the Optimal LR}
\label{sec:approx}

In this section, we extend the ideas from the 3-SAT reduction 
to show that even approximating a solution to \netlearn{} is hard.
More precisely, we show hardness for the optimization version of the problem, defined in \Cref{sec:problem} as \netlearnopt{}.
We give a proof for the Bayesian learning rule; we believe for the proof for majority vote to be similar.

\begin{theorem}[ ]\label{thm:approx}
    \netlearnopt{} with the Bayesian inference $\mu = \mu^B$ is \apx-hard. 
\end{theorem}

\begin{proof}
    We perform a PTAS reduction to \maxsat, where again each clause has exactly three literals (also referred to as \problem{Max E3-SAT}), which is known to be \apx-hard.
    In particular, we show that an approximation scheme of $ \olr(\network) $ implies an approximation scheme of \maxsat{}.

    Given an instance of \maxsat, we construct the graph $ G_\varphi $ from \Cref{def:bayesian_graph}.
    Let $ \asg^* $ be a maximal assignment, defined in \Cref{def:max_assignment}, satisfying $ M^* $ clauses.
    To prove the theorem, it suffices to show that for every $ \delta \in \left( 0,1 \right) $, there is an $ \alpha \in \left( 0,1 \right) $ such that if an ordering $ \sigma $ achieves a LR which is $ \alpha  $-close to the optimum, then the induced assignment $ \asg(\sigma) $ satisfies $ \delta M^*$ clauses.
    Abusing notation, for an assignment $ \asg $ we denote $ \clr(\asg) \deq \clr(\network,\sigma^*(\asg)) $.
    
    Note that for any 3-CNF formula, assigning truth values independently and uniformly at random satisfies $\frac 7 8$ of the clauses in expectation, implying that $ M^* \geq \frac 78 M $~\cite{Hastad2001-fg}. 
    From \Cref{lemma:bayesian_clause}, it follows that unsatisfied clauses receive optimal CLR $ \pzero $, and satisfied clauses receive at least $ \pone $ and at most $ \pthree $.
    It then follows that \begin{align*}
         \clr(\asg^*)&\geq M^*\pone+(M-M^*)\pzero+N \pcell.
    \end{align*}
    Let $ \asg' $ be an assignment which achieves $ \lr(\asg') \geq \olr(\network) - \varepsilon \geq \lr(\asg^*)  - \varepsilon $ for some $ \varepsilon > 0 $ specified later.
    Multiplying by the number of vertices in $ G_\varphi $, we get \[
        \clr(\asg') \geq \clr(\asg^*) - \varepsilon (3M+3N) \geq \clr(\asg^*) - 6\varepsilon M,
    \]
    where the second inequality holds since WLOG $ N \leq M $, by duplication of clauses.
    Re-arranging, we have that the number of satisfied clauses under $\asg'$ is
    $$M'\geq \tfrac{\pone - \pzero}{\pthree - \pzero}M^*-\tfrac{6\eps M}{\pthree - \pzero} \geq M^*\left( \tfrac{\pone - \pzero}{\pthree - \pzero}-\tfrac{48\eps}{7(\pthree - \pzero)}\right), $$
    where the second inequality holds since $M^*\geq \frac 78 M$.
    If we now set the final expression equal to $ \delta M^* $, we get an equality in $ \varepsilon $ and $ p $.
    Solving for $ \varepsilon $, we get
    \begin{equation}\label{eq:eps}
        \eps=\tfrac{7}{48}( (\pone - \pzero)-\delta(\pthree - \pzero)).
    \end{equation}
    This is a function of $ p $ and $ \delta $. 
    We only require that this $\varepsilon > 0$.
    Using that $p \in (\frac 12,1)$, this gives us \[
    \delta < \tfrac{4 p^4-8 p^3-4 p^2+8 p+2}{4 p^8-16 p^7+13 p^6+17 p^5-12 p^4-23 p^3+5 p^2+12 p+2}.
    \]
    On $p \in (\frac 12, 1)$, the RHS can be strictly lower-bounded by $p^2$. It thus suffices to set $p \geq \sqrt \delta$.
    
    This proves that an approximation to \netlearnopt{} within an additive bound $ \varepsilon $ for any $ p \geq \sqrt \delta$ implies an approximation to \maxsat{} within a factor of $\delta$. 
    Converting $\varepsilon$ to a multiplicative bound, since $ \olr \in (0,1) $, \[
        \olr - \varepsilon = \left( 1- \tfrac{\varepsilon}{\olr} \right)\olr \leq \left( 1-\varepsilon \right) \olr.
    \]
    Thus we also have a multiplicative bound $ \alpha = (1-\varepsilon) $ for each $ \delta $, which is what we wanted to prove. 
\end{proof}

\begin{corollary}[ ]
    For any fixed constant $ p \in (\sqrt{ {7}/{8}}, 1) $, an $ \alpha $-approximation of \netlearnopt{} with Bayesian inference rule is \np-hard for every $ \alpha > \alpha(p) $, defined as \[
        \alpha(p) \deq 1-\tfrac{7}{48}( (\pone - \pzero)-\tfrac 78(\pthree - \pzero)).
    \]
\end{corollary}

\begin{proof}
    This follows from \Cref{thm:approx}, since approximating \maxsat{} is \np-hard for $ \frac 78 + \xi $, where $ \xi > 0 $ \cite{Hastad2001-fg}.
    The condition of $ \alpha(p) $ is achieved by substituting $ \delta = \frac 78 $ to \Cref{eq:eps}.
\end{proof}

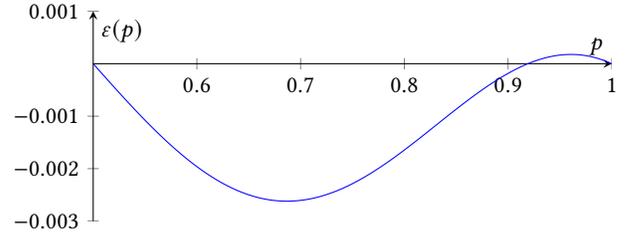
\begin{figure}[t!]
\centering
\begin{tikzpicture}
	\begin{axis}[
    xlabel={$ p $},
    ylabel={$ \varepsilon(p) $},
    legend pos=south east,
    ymax=0.001,
    ymin=-0.003,
    axis lines=center,
    scaled ticks=false,
    yticklabel style={
	/pgf/number format/precision=3,
	/pgf/number format/fixed,
    }
]

\addplot [
    domain=1/2:1, 
    samples=\figSamples,
    color=blue,
    smooth
]
 {(7*x*(-2+26*x+3*x^2-258*x^3+309*x^4+273*x^5-498*x^6-147*x^7+546*x^8-308*x^9+56*x^10))/1536};

\end{axis}
\end{tikzpicture}
\caption{The value of $ \varepsilon $ w.r.t. choice of $p$ below which (additive) approximation is hard (that is, $ \delta = \frac 78 $ \cite{Hastad2001-fg}). The requirement of $\varepsilon > 0$ gives us that $p \geq \sqrt{7/8}$.}
\Description{A plot of $\varepsilon$, initially zero, then negative, and in the interval $(\sqrt{7/8}, 1)$ slightly positive.}
\label{fig:jie_epsilon}
\end{figure}

\section{Conclusion}

In this paper, we tackle the \emph{complexity} of judging how well-equipped a given network is for social truth learning.
We consider this question in the setting of sequential decision-making by agents with bounded belief.
We then show that it is \np-hard to decide whether a large proportion of the network successfully learns a binary ground truth, both when agents are fully rational and when agents have bounded rationality. Finally, we show that it is actually hard to even approximate the learning rate of a fully rational network.

\subsection*{Future Work}

There are many remaining open directions for future work.
A natural question to ask is whether this problem belongs to the class \np, or not. Namely, it remains open whether there exists some efficiently verifiable characterization of networks which achieve high learning rates.
We conjecture that the answer is no, but we have so far been unable to prove it.

Yet another interesting direction is that of finding connections between \netlearn{} and other combinatorial problems.
\netlearn{} seems to be somewhat connected to finding big independent sets, which can allow for successful aggregation of information, along with big enough cliques for efficient spread of information. Clarifying the balance between these two contributing factors could improve our understanding of both network learning and its relation to the rest of combinatorics.

Finally, it would be fascinating to look for classes of networks for which \netlearn{} is easy.
For example, the problem of graph coloring is \np-hard in general, but easy for certain classes of graphs, such as perfect graphs. Similarly, it would be useful to develop characterizations of graphs which achieve high network learning rates and graphs whose learning rates remain bounded away from $1$. An alternative quantity of interest is the network learning rate under random, rather than optimal, decision orderings. Further, one can also consider robustness of different networks under networks with adversarial agents. 
We suspect that finding such characterizations and families of networks would give us a better understanding of the previously mentioned question of connecting this problem to the rest of the field of combinatorics.



\begin{acks}
  \Acknowledgments{}
\end{acks}



\bibliographystyle{ACM-Reference-Format} 
\bibliography{bibliography.bib}


\begin{thebibliography}{26}


\ifx \showCODEN    \undefined \def \showCODEN     #1{\unskip}     \fi
\ifx \showDOI      \undefined \def \showDOI       #1{#1}\fi
\ifx \showISBNx    \undefined \def \showISBNx     #1{\unskip}     \fi
\ifx \showISBNxiii \undefined \def \showISBNxiii  #1{\unskip}     \fi
\ifx \showISSN     \undefined \def \showISSN      #1{\unskip}     \fi
\ifx \showLCCN     \undefined \def \showLCCN      #1{\unskip}     \fi
\ifx \shownote     \undefined \def \shownote      #1{#1}          \fi
\ifx \showarticletitle \undefined \def \showarticletitle #1{#1}   \fi
\ifx \showURL      \undefined \def \showURL       {\relax}        \fi
\providecommand\bibfield[2]{#2}
\providecommand\bibinfo[2]{#2}
\providecommand\natexlab[1]{#1}
\providecommand\showeprint[2][]{arXiv:#2}

\bibitem[\protect\citeauthoryear{Acemoglu, Dahleh, Lobel, and
  Ozdaglar}{Acemoglu et~al\mbox{.}}{2011}]%
        {Acemoglu2011-vj}
\bibfield{author}{\bibinfo{person}{Daron Acemoglu}, \bibinfo{person}{Munther~A
  Dahleh}, \bibinfo{person}{Ilan Lobel}, {and} \bibinfo{person}{Asuman
  Ozdaglar}.} \bibinfo{year}{2011}\natexlab{}.
\newblock \showarticletitle{Bayesian Learning in Social Networks}.
\newblock \bibinfo{journal}{\emph{Rev. Econ. Stud.}} \bibinfo{volume}{78},
  \bibinfo{number}{4} (\bibinfo{date}{March} \bibinfo{year}{2011}),
  \bibinfo{pages}{1201--1236}.
\newblock


\bibitem[\protect\citeauthoryear{Acemoglu and Ozdaglar}{Acemoglu and
  Ozdaglar}{2011}]%
        {Acemoglu2011-tx}
\bibfield{author}{\bibinfo{person}{Daron Acemoglu} {and}
  \bibinfo{person}{Asuman Ozdaglar}.} \bibinfo{year}{2011}\natexlab{}.
\newblock \showarticletitle{Opinion dynamics and learning in social networks}.
\newblock \bibinfo{journal}{\emph{Dyn. Games Appl.}} \bibinfo{volume}{1},
  \bibinfo{number}{1} (\bibinfo{date}{March} \bibinfo{year}{2011}),
  \bibinfo{pages}{3--49}.
\newblock


\bibitem[\protect\citeauthoryear{Arieli, Sandomirskiy, and Smorodinsky}{Arieli
  et~al\mbox{.}}{2021}]%
        {arieli2020social}
\bibfield{author}{\bibinfo{person}{Itai Arieli}, \bibinfo{person}{Fedor
  Sandomirskiy}, {and} \bibinfo{person}{Rann Smorodinsky}.}
  \bibinfo{year}{2021}\natexlab{}.
\newblock \showarticletitle{On Social Networks That Support Learning}. In
  \bibinfo{booktitle}{\emph{Proceedings of the 22nd ACM Conference on Economics
  and Computation}} (Budapest, Hungary) \emph{(\bibinfo{series}{EC '21})}.
  \bibinfo{pages}{95–96}.
\newblock
\showISBNx{9781450385541}


\bibitem[\protect\citeauthoryear{Bahar, Arieli, Smorodinsky, and
  Tennenholtz}{Bahar et~al\mbox{.}}{2020}]%
        {Bahar2020-am}
\bibfield{author}{\bibinfo{person}{Gal Bahar}, \bibinfo{person}{Itai Arieli},
  \bibinfo{person}{Rann Smorodinsky}, {and} \bibinfo{person}{Moshe
  Tennenholtz}.} \bibinfo{year}{2020}\natexlab{}.
\newblock \showarticletitle{Multi-issue social learning}.
\newblock \bibinfo{journal}{\emph{Math. Soc. Sci.}}  \bibinfo{volume}{104}
  (\bibinfo{date}{March} \bibinfo{year}{2020}), \bibinfo{pages}{29--39}.
\newblock


\bibitem[\protect\citeauthoryear{Banerjee}{Banerjee}{1992}]%
        {Banerjee1992-ra}
\bibfield{author}{\bibinfo{person}{A~V Banerjee}.}
  \bibinfo{year}{1992}\natexlab{}.
\newblock \showarticletitle{A simple model of herd behavior}.
\newblock \bibinfo{journal}{\emph{Q. J. Econ.}} \bibinfo{volume}{107},
  \bibinfo{number}{3} (\bibinfo{date}{Aug.} \bibinfo{year}{1992}),
  \bibinfo{pages}{797--817}.
\newblock


\bibitem[\protect\citeauthoryear{Bikhchandani, Hirshleifer, and
  Welch}{Bikhchandani et~al\mbox{.}}{1992}]%
        {Bikhchandani1992-rs}
\bibfield{author}{\bibinfo{person}{Sushil Bikhchandani}, \bibinfo{person}{David
  Hirshleifer}, {and} \bibinfo{person}{Ivo Welch}.}
  \bibinfo{year}{1992}\natexlab{}.
\newblock \showarticletitle{A Theory of Fads, Fashion, Custom, and Cultural
  Change as Informational Cascades}.
\newblock \bibinfo{journal}{\emph{J. Polit. Econ.}} \bibinfo{volume}{100},
  \bibinfo{number}{5} (\bibinfo{date}{Oct.} \bibinfo{year}{1992}),
  \bibinfo{pages}{992--1026}.
\newblock


\bibitem[\protect\citeauthoryear{Chamley}{Chamley}{2004}]%
        {Chamley2004-or}
\bibfield{author}{\bibinfo{person}{Christophe Chamley}.}
  \bibinfo{year}{2004}\natexlab{}.
\newblock \bibinfo{booktitle}{\emph{Rational Herds: Economic Models of Social
  Learning}}.
\newblock \bibinfo{publisher}{Cambridge University Press}.
\newblock


\bibitem[\protect\citeauthoryear{Cooper}{Cooper}{1990}]%
        {Cooper1990-fn}
\bibfield{author}{\bibinfo{person}{Gregory~F Cooper}.}
  \bibinfo{year}{1990}\natexlab{}.
\newblock \showarticletitle{The computational complexity of probabilistic
  inference using bayesian belief networks}.
\newblock \bibinfo{journal}{\emph{Artif. Intell.}} \bibinfo{volume}{42},
  \bibinfo{number}{2-3} (\bibinfo{date}{March} \bibinfo{year}{1990}),
  \bibinfo{pages}{393--405}.
\newblock


\bibitem[\protect\citeauthoryear{Dagum and Luby}{Dagum and Luby}{1993}]%
        {Dagum1993-gd}
\bibfield{author}{\bibinfo{person}{Paul Dagum} {and} \bibinfo{person}{Michael
  Luby}.} \bibinfo{year}{1993}\natexlab{}.
\newblock \showarticletitle{Approximating probabilistic inference in Bayesian
  belief networks is {NP}-hard}.
\newblock \bibinfo{journal}{\emph{Artif. Intell.}} \bibinfo{volume}{60},
  \bibinfo{number}{1} (\bibinfo{date}{March} \bibinfo{year}{1993}),
  \bibinfo{pages}{141--153}.
\newblock


\bibitem[\protect\citeauthoryear{Golub and Sadler}{Golub and Sadler}{2016}]%
        {Golub2017-qo}
\bibfield{author}{\bibinfo{person}{Benjamin Golub} {and} \bibinfo{person}{Evan
  Sadler}.} \bibinfo{year}{2016}\natexlab{}.
\newblock \showarticletitle{Learning in social networks}.
\newblock In \bibinfo{booktitle}{\emph{The Oxford Handbook of the Economics of
  Networks}}. \bibinfo{publisher}{Oxford University Press},
  \bibinfo{pages}{504–542}.
\newblock
\showISBNx{9780199948277}


\bibitem[\protect\citeauthoryear{H\k{a}z\l{}a, Jadbabaie, Mossel, and
  Rahimian}{H\k{a}z\l{}a et~al\mbox{.}}{2019}]%
        {hazla2019reasoning}
\bibfield{author}{\bibinfo{person}{Jakub H\k{a}z\l{}a}, \bibinfo{person}{Ali
  Jadbabaie}, \bibinfo{person}{Elchanan Mossel}, {and}
  \bibinfo{person}{Mohammad~Ali Rahimian}.} \bibinfo{year}{2019}\natexlab{}.
\newblock \showarticletitle{Reasoning in Bayesian opinion exchange networks is
  PSPACE-hard}. In \bibinfo{booktitle}{\emph{Conference on Learning Theory}}.
  PMLR, \bibinfo{pages}{1614--1648}.
\newblock
\urldef\tempurl%
\url{https://proceedings.mlr.press/v99/hazla19a.html}
\showURL{%
\tempurl}


\bibitem[\protect\citeauthoryear{Håstad}{Håstad}{2001}]%
        {Hastad2001-fg}
\bibfield{author}{\bibinfo{person}{Johan Håstad}.}
  \bibinfo{year}{2001}\natexlab{}.
\newblock \showarticletitle{Some optimal inapproximability results}.
\newblock \bibinfo{journal}{\emph{J. ACM}} \bibinfo{volume}{48},
  \bibinfo{number}{4} (\bibinfo{date}{July} \bibinfo{year}{2001}),
  \bibinfo{pages}{798--859}.
\newblock


\bibitem[\protect\citeauthoryear{Hązła, Jadbabaie, Mossel, and
  Rahimian}{Hązła et~al\mbox{.}}{2021}]%
        {Hazla2021-vf}
\bibfield{author}{\bibinfo{person}{Jan Hązła}, \bibinfo{person}{Ali
  Jadbabaie}, \bibinfo{person}{Elchanan Mossel}, {and} \bibinfo{person}{M~Amin
  Rahimian}.} \bibinfo{year}{2021}\natexlab{}.
\newblock \showarticletitle{Bayesian Decision Making in Groups is Hard}.
\newblock \bibinfo{journal}{\emph{Oper. Res.}} \bibinfo{volume}{69},
  \bibinfo{number}{2} (\bibinfo{date}{March} \bibinfo{year}{2021}),
  \bibinfo{pages}{632--654}.
\newblock


\bibitem[\protect\citeauthoryear{Jadbabaie, Molavi, Sandroni, and
  Tahbaz-Salehi}{Jadbabaie et~al\mbox{.}}{2012}]%
        {Jadbabaie2012-ob}
\bibfield{author}{\bibinfo{person}{Ali Jadbabaie}, \bibinfo{person}{Pooya
  Molavi}, \bibinfo{person}{Alvaro Sandroni}, {and} \bibinfo{person}{Alireza
  Tahbaz-Salehi}.} \bibinfo{year}{2012}\natexlab{}.
\newblock \showarticletitle{Non-Bayesian social learning}.
\newblock \bibinfo{journal}{\emph{Games Econ. Behav.}} \bibinfo{volume}{76},
  \bibinfo{number}{1} (\bibinfo{date}{Sept.} \bibinfo{year}{2012}),
  \bibinfo{pages}{210--225}.
\newblock


\bibitem[\protect\citeauthoryear{Kwisthout}{Kwisthout}{2018}]%
        {Kwisthout2018-az}
\bibfield{author}{\bibinfo{person}{Johan Kwisthout}.}
  \bibinfo{year}{2018}\natexlab{}.
\newblock \showarticletitle{Approximate inference in Bayesian networks:
  Parameterized complexity results}.
\newblock \bibinfo{journal}{\emph{Int. J. Approx. Reason.}}
  \bibinfo{volume}{93} (\bibinfo{date}{Feb.} \bibinfo{year}{2018}),
  \bibinfo{pages}{119--131}.
\newblock


\bibitem[\protect\citeauthoryear{Laland}{Laland}{2004}]%
        {Laland2004-ej}
\bibfield{author}{\bibinfo{person}{Kevin~N Laland}.}
  \bibinfo{year}{2004}\natexlab{}.
\newblock \showarticletitle{Social learning strategies}.
\newblock \bibinfo{journal}{\emph{Learn. Behav.}} \bibinfo{volume}{32},
  \bibinfo{number}{1} (\bibinfo{date}{Feb.} \bibinfo{year}{2004}),
  \bibinfo{pages}{4--14}.
\newblock


\bibitem[\protect\citeauthoryear{Lu, Chong, Lu, and Gao}{Lu
  et~al\mbox{.}}{2024}]%
        {lu24enabling}
\bibfield{author}{\bibinfo{person}{Kevin Lu}, \bibinfo{person}{Jordan Chong},
  \bibinfo{person}{Matt Lu}, {and} \bibinfo{person}{Jie Gao}.}
  \bibinfo{year}{2024}\natexlab{}.
\newblock \showarticletitle{Enabling Asymptotic Truth Learning in a Social
  Network}. In \bibinfo{booktitle}{\emph{Proceedings of the 20th Conference on
  Web and Internet Economics (WINE'24)}}.
\newblock


\bibitem[\protect\citeauthoryear{Mobius and Rosenblat}{Mobius and
  Rosenblat}{2014}]%
        {Mobius2014-oy}
\bibfield{author}{\bibinfo{person}{Markus Mobius} {and} \bibinfo{person}{Tanya
  Rosenblat}.} \bibinfo{year}{2014}\natexlab{}.
\newblock \showarticletitle{Social Learning in Economics}.
\newblock \bibinfo{journal}{\emph{Annu. Rev. Econom.}} \bibinfo{volume}{6},
  \bibinfo{number}{1} (\bibinfo{date}{Aug.} \bibinfo{year}{2014}),
  \bibinfo{pages}{827--847}.
\newblock


\bibitem[\protect\citeauthoryear{Mossel, Neeman, and Tamuz}{Mossel
  et~al\mbox{.}}{2014}]%
        {Mossel2014-mv}
\bibfield{author}{\bibinfo{person}{Elchanan Mossel}, \bibinfo{person}{Joe
  Neeman}, {and} \bibinfo{person}{Omer Tamuz}.}
  \bibinfo{year}{2014}\natexlab{}.
\newblock \showarticletitle{Majority dynamics and aggregation of information in
  social networks}.
\newblock \bibinfo{journal}{\emph{Auton. Agent. Multi. Agent. Syst.}}
  \bibinfo{volume}{28}, \bibinfo{number}{3} (\bibinfo{date}{May}
  \bibinfo{year}{2014}), \bibinfo{pages}{408--429}.
\newblock


\bibitem[\protect\citeauthoryear{Mossel and Tamuz}{Mossel and Tamuz}{2017}]%
        {Mossel2017-sd}
\bibfield{author}{\bibinfo{person}{Elchanan Mossel} {and} \bibinfo{person}{Omer
  Tamuz}.} \bibinfo{year}{2017}\natexlab{}.
\newblock \showarticletitle{Opinion exchange dynamics}.
\newblock \bibinfo{journal}{\emph{Probab. Surv.}} \bibinfo{volume}{14},
  \bibinfo{number}{none} (\bibinfo{date}{Jan.} \bibinfo{year}{2017}),
  \bibinfo{pages}{155--204}.
\newblock


\bibitem[\protect\citeauthoryear{Sgroi}{Sgroi}{2002}]%
        {Sgroi2002-rz}
\bibfield{author}{\bibinfo{person}{Daniel Sgroi}.}
  \bibinfo{year}{2002}\natexlab{}.
\newblock \showarticletitle{Optimizing Information in the Herd: Guinea Pigs,
  Profits, and Welfare}.
\newblock \bibinfo{journal}{\emph{Games Econ. Behav.}} \bibinfo{volume}{39},
  \bibinfo{number}{1} (\bibinfo{date}{April} \bibinfo{year}{2002}),
  \bibinfo{pages}{137--166}.
\newblock


\bibitem[\protect\citeauthoryear{Shoham and Tennenholtz}{Shoham and
  Tennenholtz}{1992}]%
        {Shoham1992-ir}
\bibfield{author}{\bibinfo{person}{Yoav Shoham} {and} \bibinfo{person}{Moshe
  Tennenholtz}.} \bibinfo{year}{1992}\natexlab{}.
\newblock \showarticletitle{Emergent conventions in multi-agent systems:
  initial experimental results and observations (preliminary report)}. In
  \bibinfo{booktitle}{\emph{Proceedings of the Third International Conference
  on Principles of Knowledge Representation and Reasoning}}
  \emph{(\bibinfo{series}{KR'92})}. \bibinfo{pages}{225--231}.
\newblock


\bibitem[\protect\citeauthoryear{Smith}{Smith}{1991}]%
        {Smith1991-sy}
\bibfield{author}{\bibinfo{person}{Lones Smith}.}
  \bibinfo{year}{1991}\natexlab{}.
\newblock \emph{\bibinfo{title}{Essays on dynamic models of equilibrium and
  learning}}.
\newblock \bibinfo{thesistype}{Ph.D. Dissertation}. \bibinfo{school}{University
  of Chicago}.
\newblock


\bibitem[\protect\citeauthoryear{Smith and S{\o}rensen}{Smith and
  S{\o}rensen}{2000}]%
        {Smith2000-wk}
\bibfield{author}{\bibinfo{person}{Lones Smith} {and} \bibinfo{person}{Peter
  S{\o}rensen}.} \bibinfo{year}{2000}\natexlab{}.
\newblock \showarticletitle{Pathological Outcomes of Observational Learning}.
\newblock \bibinfo{journal}{\emph{Econometrica}} \bibinfo{volume}{68},
  \bibinfo{number}{2} (\bibinfo{date}{March} \bibinfo{year}{2000}),
  \bibinfo{pages}{371--398}.
\newblock


\bibitem[\protect\citeauthoryear{Wang, Luo, and Gao}{Wang
  et~al\mbox{.}}{2022}]%
        {Wang2022-uy}
\bibfield{author}{\bibinfo{person}{Haotian Wang}, \bibinfo{person}{Feng Luo},
  {and} \bibinfo{person}{Jie Gao}.} \bibinfo{year}{2022}\natexlab{}.
\newblock \showarticletitle{Co-evolution of opinion and social tie dynamics
  towards structural balance}.
\newblock In \bibinfo{booktitle}{\emph{Proceedings of the 2022 Annual ACM-SIAM
  Symposium on Discrete Algorithms (SODA)}}. \bibinfo{publisher}{Society for
  Industrial and Applied Mathematics}, \bibinfo{address}{Philadelphia, PA},
  \bibinfo{pages}{3362--3388}.
\newblock


\bibitem[\protect\citeauthoryear{Welch}{Welch}{1992}]%
        {Welch1992-yt}
\bibfield{author}{\bibinfo{person}{Ivo Welch}.}
  \bibinfo{year}{1992}\natexlab{}.
\newblock \showarticletitle{Sequential sales, learning, and cascades}.
\newblock \bibinfo{journal}{\emph{J. Finance}} \bibinfo{volume}{47},
  \bibinfo{number}{2} (\bibinfo{date}{June} \bibinfo{year}{1992}),
  \bibinfo{pages}{695--732}.
\newblock


\end{thebibliography}

\clearpage
\counterwithin{figure}{section}
\onecolumn
\appendix

\section{Bayesian General Learning Rates}
\label{app:bayesian_gadgets}

\begin{lemma}[ ]\label{lemma:twoInformedOneUninformed}
    Let $ \frac 12 < p < 1 $.
    Let $ p' = \frac p2 + \frac 32 p^2 - p^3 $.
    Then \[
        \left( \frac p{1-p} \right)^2 \frac {1-p'}{p'} > 1.
    \]
\end{lemma}

\begin{proof}
    This Lemma is proven by solving the stated inequality.
    We check all our calculations using Wolfram Mathematica.
\end{proof}

\ourparagraph{General Clause Setup}
Consider a clause gadget with literal nodes $a$, $b$, and $c$.
WLOG, suppose the optimal ordering $\sigma^*$ orders $a$ before $b$ before $c$.
Let $\alpha$, $\beta$, and $\gamma$ be the literals corresponding to $ a $, $ b $, and $ c $, respectively.
We also use $\alpha$, $\beta$, and $\gamma$, to denote the corresponding nodes from their respective cells. (see \Cref{fig:bayesian_alphabetagamma} for details of the construction).
Suppose that $\alpha$, $\beta$, and $\gamma$ correctly predict the ground truth with probabilities $p_\alpha$, $p_\beta$, and $p_\gamma$, respectively.
Note that since the literal nodes have no outgoing edges except to each other, ordering all literal nodes after variable nodes gives the literal nodes more information at no cost to any other nodes.
So we assume $ \sigma^* $ orders all literal nodes after variable nodes, and also $a$ before $b$ before $c$. 

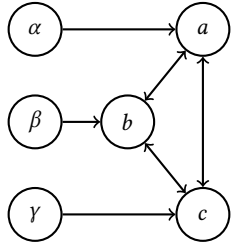
\begin{figure}[t!]
	\centering
	\begin{tikzpicture}[
		xsh/.style = { xshift=10mm },
		node distance = 5mm and 5mm,
	stff/.style={circle, draw=black, \figThickness, minimum size=\figCircSize, inner sep=0pt},
	]
		\node[stff]        (alpha)                  {$ \alpha $};
		\node[stff]        (beta)    [below=of alpha]  {$ \beta $};
		\node[stff]        (gamma)    [below=of beta]  {$ \gamma $};

		\node[stff]        (a)    [right=of alpha,xsh]  {$ a $};
		\node[stff]        (b)    [right=of beta]  {$ b $};
		\node[stff]        (c)    [right=of gamma,xsh]  {$ c $};

		\draw[->, \figThickness] (alpha)  to (a);
		\draw[->, \figThickness] (beta)  to (b);
		\draw[->, \figThickness] (gamma)  to (c);

		\draw[<->, \figThickness] (a)  to (c);
		\draw[<->, \figThickness] (b)  to (c);
		\draw[<->, \figThickness] (b)  to (a);
	\end{tikzpicture}
	\caption{Construction used in General Ordering Learning Rate proof. $\alpha, \beta$ and $\gamma$ are the vertices from the cell gadgets corresponding to the literals from $C$.}
 \Description{A graph of the Bayes gadget  with literals $a,b,c$, along with $alpha,beta,gamma$, from which edges go only to $a,b,c,$ respectively.}
     \label{fig:bayesian_alphabetagamma}
\end{figure}

\begin{observation}
    The predictions of $ \alpha $, $ \beta $, $ \gamma $ are independent.
\end{observation}

\begin{proof}
    The nodes $ \alpha $, $ \beta $, $ \gamma $ are from their variable cells.
    From each cell, there are only outward edges.
    This means that nothing outside of a cell can affect any of its nodes.
    This includes the nodes of other cells.
    Hence, the predictions are independent.
\end{proof}

\begin{lemma}\label{lemma:bayesian_l1_LR}
    For $a$ the earliest literal node in its clause gadget, $\oclr(a) = p$ when its corresponding literal $\alpha$ is off, and $\oclr(a) = p' \deq \frac p2 + \frac 32 p^2 - p^3$ when $\alpha$ is on. 
\end{lemma}
\begin{proof}
    We compute the learning rate of $ a $ under Bayesian aggregation.
    Recall that Bayesian agents base their predictions on the posterior probability for $\theta$ given their inputs, tie-breaking by selecting an action randomly.
    So, the prediction of node $a$ is determined by the posterior of $\theta$ given their private signal $s_a$ and their in-neighbor $\alpha$.
    Therefore, $a$ predicts correctly whenever its posterior for $\theta = 1$ exceeds $1/2$ or tiebreaks correctly.
    Equivalently, $a$ predicts correctly whenever its corresponding likelihood ratio exceeds $1$ or tiebreaks correctly.
    We denote the ratio as $ \Lambda $, defined as
    \[
        \Lambda(\theta = 1 : \theta = 0 \mid s_a, \alpha) \deq \frac {\Pr[s_a, \alpha \mid \theta = 1]} {\Pr[s_a = 1, \alpha = 0 \mid \theta = 0]}                                                            = \frac {\Pr[s_a \mid \theta = 1]} {\Pr[s_a \mid \theta = 0]} \cdot \frac {\Pr[ \alpha \mid \theta = 1]} {\Pr[\alpha \mid \theta = 0]},
    \]
    where the equality holds due to $ s_a $ and $ \alpha $'s prediction being independent.

    In general, there are four cases which may arise.
    \[
                \Lambda(1:0 \mid s_a, \alpha) = \begin{cases}
                     \frac p{1-p} \cdot \frac {p_\alpha}{1-p_\alpha} > 1 & \text{if $s_a = 1, \alpha = 1$,} \\
                     \frac {1-p}p \cdot \frac {1-{p_\alpha}}{p_\alpha} < 1 & \text{if $s_a = 0, \alpha = 0$,} \\
                     \frac {1-p}p\cdot \frac{p_\alpha}{1-p_\alpha} & \text{if $s_a = 0, \alpha = 1$,} \\
                     \frac p{1-p} \cdot\frac{1-p_\alpha}{p_\alpha} & \text{if $s_a = 1, \alpha = 0$.} \\
                \end{cases}
    \]
    The inequality in the first two cases holds because $ p_\alpha \geq p $.
    Cases 3 and 4 need to be analyzed separately for $ \alpha $ on and off.
    \begin{itemize}
        \item If $\alpha$ is off, then $ p_\alpha = \pr{\alpha = 1 \mid  \theta = 1} = \pr{\alpha = 0 \mid \theta = 0} = p $, as discussed in \Cref{lemma:bayesian_cellLearningRate}.
            Thus \begin{align*}
                \Lambda(1:0 \mid s_a = 0, \alpha=1) = 
                \Lambda(1:0 \mid s_a = 1, \alpha=0)                 = \frac p{1-p} \cdot\frac{1-p}{p} = 1.
            \end{align*}
        \item If $\alpha$ is on, then $ p_\alpha = \pr{\alpha = 1 \mid  \theta = 1} = \pr{\alpha = 0 \mid \theta = 0} = \frac p2 + \frac 32 p^2 - p^3 > p $, again, as discussed in \Cref{lemma:bayesian_cellLearningRate}.
            Thus $\Lambda(1:0 \mid s_a = 0, \alpha=1) > 1$ and $\Lambda(1:0 \mid s_a = 1, \alpha=0) < 1.$ \end{itemize}
    
            Now, we compute the final learning rate of $ a $, which is the probability $ \pr{a=\theta} $. Since we assume $ q = \frac 12 $, this is the same as $ \pr{a=1 \suchthat \theta=1} $.
    
	Therefore, the probability of $a$ being correct is the probability of any input configuration yielding either $\Lambda > 1$, plus the probability of a configuration yielding $\Lambda = 1$ and tie breaking successfully with probability $ \frac 12 $.
    \begin{align*}
        \Pr[a = 1 \mid \alpha \text{ off}\,] &= \Pr[s_a = \alpha = 1 \mid \theta = 1] + \frac 12 \pr{s_a \neq \alpha \mid \theta = 1} = p^2 + \frac 12 \left( p(1-p) + (1-p)p \right) = p, \\
            \Pr[a = 1 \mid \alpha \text{ on}\,] &= \Pr[\alpha = 1 \mid \theta = 1] = p_\alpha = \frac p2 + \frac 32 p^2 - p^3.
	\end{align*}
\end{proof}

\begin{lemma}\label{lemma:bayesian_l2_LR}
    For $b$ the second earliest literal node in its clause gadget, $\oclr(b)$ takes the following values:
    \begin{align*}
        \oclr(b) = \begin{cases}
        3p^2-2p^3 & \text{if $ \alpha $, $ \beta $ both off,} \\
        p'(p'+2p-2pp') & \text{if $ \alpha $, $ \beta $ both on,} \\
        p(p+2p'-2pp') & \text{otherwise,} \\
    \end{cases}
    \end{align*}
    where $p' = \frac p2 + \frac 32 p^2 - p^3 $.
\end{lemma}
\begin{proof}
    As in the previous Lemmas, we first compute the ratio of $ \frac{\pr{\theta=1 \mid X_b}}{\pr{\theta=0 \mid X_b}} $, denoted as $ \Lambda $, to get $ b $'s predictions for each of its input configurations. We then compute the probability that $b$ sees inputs that yield a correct prediction.

    The information available to $ b $ is the predictions of $ a $ and $ \beta $, along with $ s_b $, $ b $'s private signal.
    Observe that these are all independent.
    We again define the ratio $ \Lambda $ as \begin{align*}
        \Lambda &\deq \frac {\Pr[s_b, a, \beta \mid \theta = 1]} {\Pr[s_b, a, \beta \mid \theta = 0]} = \frac {\Pr[s_b \mid \theta = 1]} {\Pr[s_b \mid \theta = 0]} \cdot \frac {\Pr[ a \mid \theta = 1]} {\Pr[a \mid \theta = 0]}\cdot\frac {\Pr[ \beta \mid \theta = 1]} {\Pr[\beta \mid \theta = 0]}.
    \end{align*}
    The values of $ \Lambda $ are once again determined by $ p \deq \pr{s_b = \theta} $ and $ p_a \deq \pr{a = \theta} $, $ p_\beta = \pr{\beta = \theta} $.
    Notice also that $ p_a, p_\beta \geq p $. We now list $\Lambda$ for all possible inputs.

    \begin{enumerate}[ref=(\theenumi)]
        \item If $ s_b = a = \beta = 1 $, \[
            \Lambda = \frac p{1-p} \cdot \frac {p_a}{1-p_a} \cdot \frac {p_\beta}{1-p_\beta} > 1
        \]
    \item\label{badcase1} If $ s_b = a = 1, \beta = 0 $, \[
                \Lambda = \frac p{1-p} \cdot \frac {p_a}{1-p_a} \cdot \frac {1-p_\beta}{p_\beta}.
            \]
        \item\label{badcase2} If $ s_b = 1,  a = 0, \beta = 1 $, \[
                \Lambda = \frac p{1-p} \cdot \frac {1-p_a}{p_a} \cdot \frac {p_\beta}{1-p_\beta}.
            \]
	\item If $ s_b = 1,  a = \beta = 0 $, \[
                \Lambda = \frac p{1-p} \cdot \frac {1-p_a}{p_a} \cdot \frac {1-p_\beta}{p_\beta} < 1.
            \]
    \end{enumerate}
    The rest of the cases follow trivially from symmetry, since
    \[
        \Lambda(s_b, a, \beta) = 1/\Lambda(\bar s_b, \bar a, \bar \beta).
    \]

    For the cases \labelcref{badcase1,badcase2}, we perform case analysis on whether $ \alpha $ and $ \beta $ are on or off.
    By \Cref{lemma:bayesian_l1_LR}, $p_a = p$ when $\alpha$ is off and $p_a = p'$ when $\alpha$ is on.
    From \Cref{lemma:bayesian_cellLearningRate}, $\beta$ is correct with probability $p$ when $\beta$ is off and $p'$ when $\beta$ is on.
    We consider two cases:
    \begin{itemize}[ ]
        \item If $ \alpha, \beta $ are both on, or both off, then $ p_a = p_\beta $, leaving $\Lambda = \frac p{1-p} > 1 $ in both of the middle cases.
        \item Otherwise, WLOG $ \alpha $ is on, $ \beta $ off.
            Then $ p_a = p', p_\beta = p $, resulting in $\Lambda( s_b= a=1, \beta=0) = \frac {p'}{1-{p'}} > 1 $ and $\Lambda( s_b= 
            \beta=1, a=0) = \left( \frac{p}{1-p} \right)^2 \cdot \frac {1-p'}{p'} > 1 $ by \Cref{lemma:twoInformedOneUninformed}.
    \end{itemize}

    We can see that $ b $ simply follows the majority of signals.
    \begin{align*}
        \pr{b = \theta} &= \pr{b = 1 \mid \theta = 1} = \pr{s_b + a + \beta \geq 2} = 1-\pr{s_b + a + \beta \leq 1} \\
                        &= 1 - (1-p)(1-p_\alpha)(1-p_\beta) - p(1-p_\alpha)(1-p_\beta) - (1-p)p_\alpha(1-p_\beta) - (1-p)(1-p_\alpha)p_\beta.
    \end{align*}
    Substituting $ p_a = p $ (similarly, $p_b = p$) if $ \alpha $ ($\beta$) is off and $ p' $ otherwise yields the desired probabilities.
\end{proof}

\begin{lemma}\label{lemma:bayesian_l3_LR}
    For $c$ the last literal node in its clause gadget, $\oclr(c)$ takes the following values, depending on the state of $ \left( \alpha, \beta, \gamma \right) $:
    \begin{align*}
\text{(off},\text{off},\text{off)}&\to            p^2 \left(2 p^3-5 p^2+2 p+2\right), \\
\text{(off},\text{off},\text{on)}&\to            p \left(-p^2-p (p'-2)+p'\right), \\
\text{(off},\text{on},\text{off)}&\to            p \left(p^3 (2 p'-1)+p^2 (1-4 p')+p (p'+1)+p'\right), \\
\text{(off},\text{on},\text{on)}&\to            p \left(2 p^2 (p'-1) p'+p \left(-3 (p')^2+p'+1\right)+p' (p'+1)\right), \\
\text{(on},\text{off},\text{off)}&\to            p^4 (2 p'-1)+p^3 (2-4 p')+2 p p', \\
\text{(on},\text{off},\text{on)}&\to            p \left(p^2 \left(2 (p')^2-2 p'+1\right)-3 p (p')^2+p' (p'+2)\right), \\
\text{(on},\text{on},\text{off)}&\to            p \left(4 p^2 (p'-1) p'+p \left(-6 (p')^2+4 p'+1\right)+2 (p')^2\right), \\
\text{(on},\text{on},\text{on)}&\to            p' \left(p^2 \left(2 (p')^2-3 p'+1\right)+p \left(-2 (p')^2+p'+1\right)+p'\right).
    \end{align*}
\end{lemma}

\begin{proof}
    We again compute the value of $\Lambda = \frac{\pr{\theta = 1 \mid X_c}}{\pr{\theta = 0 \mid X_c}}$ for all possible configurations of $ X_c $.
    From that, we determine the action $ c $ chooses and compute the learning rate.

    First, let us rewrite \[
        \Lambda = \frac{\pr{s_c \mid \theta = 1}}{\pr{s_c \mid \theta = 0}} \cdot \frac{\pr{\gamma \mid \theta = 1}}{\pr{\gamma \mid \theta = 0}} \cdot \frac{\pr{a \mid \theta = 1}}{\pr{a \mid \theta = 0}} \cdot \frac{\pr{b \mid a,\theta = 1}}{\pr{b \mid a,\theta = 0}}.
    \]

    From the proofs of the previous two lemmas, we have \begin{align*}
        \pr{b = 1 \mid a = 1, \theta = 1} &= \pr{b = 0 \mid a = 0, \theta = 0} = p + p_\beta - pp_\beta, \\
        \pr{b = 1 \mid a = 0, \theta = 1} &= \pr{b = 0 \mid a = 1, \theta = 0} = pp_\beta, \\
        \pr{a = 1 \mid \theta = 1} &= \pr{a = 0 \mid \theta = 0} = p_\alpha, \\
    \end{align*}
    where $ p_\alpha = p' $ if $ \alpha  $ is on, $ p $ otherwise.
    We again use $ p_\beta, p_\gamma  $ similarly.

    We now perform case analysis on $ \left( s_c, \gamma, a, b \right) $:
    \begin{enumerate}[ ]
        \item $ \left( 1,1,1,1 \right) $: \[
            \Lambda = \frac{p}{1-p} \cdot \frac{p_\gamma}{1-p_\gamma} \cdot \frac{p_\alpha}{1-p_\alpha} \cdot \frac{p+p_\beta - pp_\beta}{1-pp_\beta} > 1,
        \]
        since it holds that $ 1>p,p_\alpha, p_\beta, p_\gamma> \frac 12  $.
        \item $ \left( 1,1,1,0 \right) $: \begin{align*}
                \Lambda &= \frac{p}{1-p} \cdot \frac{p_\gamma}{1-p_\gamma} \cdot \frac{p_\alpha}{1-p_\alpha} \cdot \frac{\left( 1-p \right)\left( 1-p_\beta \right)}{pp_\beta} = \frac{p_\gamma}{1-p_\gamma} \cdot \frac{p_\alpha}{1-p_\alpha} \cdot \frac{1-p_\beta}{p_\beta} \geq \frac{p}{1-p} \cdot \frac{p}{1-p}  \cdot \frac{1-p'}{p'} > 1.
        \end{align*}
        \item $ \left( 1,1,0,1 \right) $: \begin{align*}
                \Lambda &= \frac{p}{1-p} \cdot \frac{p_\gamma}{1-p_\gamma} \cdot \frac{1-p_\alpha}{p_\alpha} \cdot \frac{pp_\beta}{\left( 1-p \right)\left( 1-p_\beta \right)} \geq \left( \frac{p}{1-p} \right)^4 \cdot \frac{1-p'}{p'} > 1.
        \end{align*}
        \item $ \left( 1,1,0,0 \right) $: \begin{align*}
                \Lambda &= \frac{p}{1-p} \cdot \frac{p_\gamma}{1-p_\gamma} \cdot \frac{1-p_\alpha}{p_\alpha} \cdot \frac{1-pp_\beta}{p+p_\beta - pp_\beta}.
        \end{align*}
        This is $ <1 $ if $ \alpha, \beta  $ are on and $ \gamma $ is off, and $ >1 $ otherwise. (By Wolfram, need to write this out)
        \item $ \left( 1,0,1,1 \right) $: \begin{align*}
                \Lambda &= \frac{p}{1-p} \cdot \frac{1-p_\gamma}{p_\gamma} \cdot \frac{p_\alpha}{1-p_\alpha} \cdot \frac{p+p_\beta - pp_\beta}{1-pp_\beta} \geq \left( \frac{p}{1-p} \right)^2 \cdot \frac{1-p'}{p'} \cdot \frac{p+p_\beta - pp_\beta}{1-pp_\beta} > \frac{p+p_\beta - pp_\beta}{1-pp_\beta} > 1.
        \end{align*}
        \item $ \left( 1,0,1,0 \right) $: \begin{align*}
                \Lambda &= \frac{p}{1-p} \cdot \frac{1-p_\gamma}{p_\gamma} \cdot \frac{p_\alpha}{1-p_\alpha} \cdot \frac{\left( 1-p \right)\left( 1-p_\beta \right)}{pp_\beta} \leq \frac{p'}{1-p'} \cdot \left( \frac{1-p}{p} \right)^2 <1.
        \end{align*}
        \item $ \left( 1,0,0,1 \right) $: \begin{align*}
                \Lambda &= \frac{p}{1-p} \cdot \frac{1-p_\gamma}{p_\gamma} \cdot \frac{1-p_\alpha}{p_\alpha} \cdot \frac{pp_\beta}{\left( 1-p \right)\left( 1-p_\beta \right)} \geq \left( \frac{1-p'}{p'} \right)^2 \cdot \left( \frac{p}{1-p} \right)^3.
        \end{align*}
        The lower bound on $ \Lambda $ is reached for $ \gamma $ on and $ \alpha, \beta $ off, and in such a case it is $ <1 $.
        However, in any other state of $ \alpha, \beta, \gamma $, the value of $ \Lambda $ is lower-bounded by \[
             \frac{1-p'}{p'} \cdot \left( \frac{p}{1-p} \right)^2 > 1,
        \]
        by \Cref{lemma:twoInformedOneUninformed}.
        \item $ \left( 1,0,0,0 \right) $: \begin{align*}
                \Lambda &= \frac{p}{1-p} \cdot \frac{1-p_\gamma}{p_\gamma} \cdot \frac{1-p_\alpha}{p_\alpha} \cdot \frac{1-pp_\beta}{1-\left( 1-p \right)\left( 1-p_\beta \right)} \leq \frac{1-p}{p} \cdot \frac{1-p^2}{2p-p^2}<1.
        \end{align*}
    \end{enumerate}
    The remaining cases follow from symmetry.

    We now compute the learning rates in each case, and multiply them by the probability of that case taking place.
    Due to symmetry of $ \theta $ and $ 1-\theta $, this is equivalent to $\pr{\Lambda > 1 \mid \theta = 1}$.
    This can be written as the sum of the probabilities of all cases yielding $ \Lambda > 1 $ assuming $ \theta = 1 $.

    The configurations which always yield $ \Lambda > 1 $ are $ \left( 1,1,1,1 \right)$, $(1,1,1,0)$, $(1,1,0,1)$, $(1,0,1,1)$, $(0,1,0,1)$, $(0,1,1,1) $.
    Then, if $ \alpha,\beta $ are off and $ \gamma $ is on, then $ (0,1,1,0) $ gives $ \Lambda > 1 $, otherwise $ (1,0,0,1) $ gives $ \Lambda > 1 $.
    Finally, if $ \alpha,\beta $ are on and $ \gamma $ is off, then $ (1,1,0,0) $ gives $ \Lambda > 1 $, otherwise $ (0,0,1,1) $ gives $ \Lambda > 1 $.
    The learning rate of $ c $ is then the sum of the probabilities of these cases, given that $ \theta = 1 $.
\end{proof}

We now compute the CLR in the following three cases: when all literals are off, when all literals are on, and when only one literal is on.

\begin{lemma}\label{lemma:bayesian_000CLR} 
    Suppose that an optimal ordering $\sigma^*$ sets all literals in a clause $ C $ to be off. Then the resulting CLR of $ \gadget C $ is \[
        \oclr (\gadget C) = \pzero \deq p \left(2 p^4-5 p^3+5 p+1\right).
     \]
\end{lemma}

\begin{proof}
    Follows directly from \Cref{lemma:bayesian_l1_LR,lemma:bayesian_l2_LR,lemma:bayesian_l3_LR}.
\end{proof}

\begin{lemma}
    \label{lemma:bayesian_111CLR}
    Suppose that an optimal ordering $\sigma^*$ sets all literals in a clause $ C $ to be on. Then the resulting CLR of $ \gadget C $ is \begin{align*}
        \oclr(\gadget C) = \pthree \deq p' \left(p^2 \left(2 (p')^2-3 p'+1\right)-p \left(2 (p')^2+p'-3\right)+2 p'+1\right).
     \end{align*}
\end{lemma}

\begin{proof}
    Follows again directly from \Cref{lemma:bayesian_l1_LR,lemma:bayesian_l2_LR,lemma:bayesian_l3_LR}.
    The expression in the lemma statement is just the sum of $ \oclr(a),\oclr(b),\oclr(c) $.
\end{proof}

\begin{lemma}
    \label{lemma:bayesian_100CLR}
    Suppose that an optimal ordering $\sigma^*$ sets exactly one literal to be on. Then the resulting CLR is \[
        \oclr(\gadget C) =\pone \deq p^4 (2 p'-1)+p^3 (2-4 p')+p^2 (1-2 p')+4 p p'+p'.
     \]
\end{lemma}

\begin{proof}
    This lemma is a bit more involved than the two ealrier ones, since there are now three cases which are not symmetrical---the vertex corresponding to the ``on'' literal is either first, second, or third (in other words, it corresponds to either $ a $, $ b $, or $ c $ from the general lemmas).

    It turns out that the highest learning rate is achieved when $ \alpha $ is on.
    Intuitively, the fact that $ \alpha $ is on means that $ a $ receives a stronger signal, and it is able to be spread all over the gadget.
    Formally, we show that the learning rate corresponding to $ \left( \alpha, \beta,\gamma \right) = \left( \text{on},\text{off},\text{off} \right) $ is higher than $ \left( \text{off},\text{off},\text{on} \right) $ and $ \left( \text{off},\text{on},\text{off} \right) $.

    By \Cref{lemma:bayesian_l1_LR,lemma:bayesian_l2_LR,lemma:bayesian_l3_LR}, and assuming that all the variable cells are ordered before the clause gadget, the learning rate of the clause $ C $ is as follows.
    \begin{enumerate}[ ]
        \item If $ \alpha $ is on, then \[
             \clr(\gadget C) = p^4 (2 p'-1)+p^3 (2-4 p')+p^2 (1-2 p')+4 p p'+p',
        \]
        \item if $ \beta $ is on, then \[
            \clr(\gadget C) = p \left(p^3 (2 p'-1)+p^2 (1-4 p')-p (p'-2)+3 p'+1\right),
        \]
        \item if $ \gamma $ is on, then \[
            \clr(\gadget C) = p \left(-3 p^2-p (p'-5)+p'+1\right).
        \]
    \end{enumerate}
    It is not hard to see that indeed the first quantity is, assuming $ \frac 12 > p > 1 $, the largest of the three, and it is thus the optimal value for the cumulative learning rate of a cell with only one variable turned on.
\end{proof}

\begin{lemma}
    \label{lemma:bayesian_110CLR}
    Suppose that an optimal ordering $\sigma^*$ sets exactly two literals to be on. Then the resulting CLR is \[
        \oclr(\gadget C) =\ptwo \deq 4 p^3 (p'-1) p'+p^2 \left(-6 (p')^2+4 p'+1\right)+2 p p'+p' (p'+1).
     \]
\end{lemma}

\begin{proof}
    As in the previous lemma, we need to decide which of the literals should be the ``off'' one---in this case it is $ \gamma $.
    Intuitively, again, the fact that $ \alpha $ and $ \beta $ are on means that their stronger signals can be better spread over the gadget.
    Formally, we show that the learning rate corresponding to $ \left( \alpha, \beta,\gamma \right) = \left( \text{on},\text{on},\text{off} \right) $ is the highest.

    By \Cref{lemma:bayesian_l1_LR,lemma:bayesian_l2_LR,lemma:bayesian_l3_LR}, and assuming that all the variable cells are ordered before the clause gadget, the learning rate of the clause $ C $ is as follows.
    \begin{enumerate}[ ]
        \item If $ \alpha $ is off, then \[
             \clr(\gadget C) = p \left(2 p^2 (p'-1) p'-p \left(3 (p')^2+p'-2\right)+(p')^2+3 p'+1\right),
        \]
        \item if $ \beta $ is off, then \[
            \clr(\gadget C) = p'+p p' (4+p')+p^2 (1-2 p'-3 (p')^2)+p^3 (1-2 p'+2 (p')^2),
        \]
        \item if $ \gamma $ is off, then \[
            \clr(\gadget C) = 4 p^3 (p'-1) p'+p^2 \left(-6 (p')^2+4 p'+1\right)+2 p p'+p' (p'+1).
        \]
    \end{enumerate}
    It is not hard to see that indeed the last quantity is, assuming $ \frac 12 > p > 1 $, the largest of the three, and it is thus the optimal value for the cumulative learning rate of a cell with exactly two variables turned on.
\end{proof}

Lastly, observe that as the number of literals which are set to on increases, the optimal gadget learning rate increases because the literal nodes get stronger signals for $\theta$. We can verify this for the cases computed above:
\begin{corollary}\label{cor:bayesian_LR_comparison}
    For $p \in (\frac 1 2, 1)$, \[
    \pthree \geq \ptwo \geq \pone \geq \pzero.
    \]
\end{corollary}

\begin{proof}
    Follows by substituting the expressions from \Cref{lemma:bayesian_000CLR,lemma:bayesian_111CLR,lemma:bayesian_100CLR,lemma:bayesian_110CLR}.
    The calculations were verified by Wolfram Mathematica.
    For a visual representation of these learning rates, see \Cref{fig:bayes_ps}.
\end{proof}

\begin{corollary} \label{obs:satisfied_bounds}
    For clause gadgets with at least one on literal, the optimal gadget learning rate is lower bounded by $\pone$ and upper bounded by $\pthree$.
\end{corollary} 

\section{Full Proof of Theorem~\ref{thm:nphardness_maj}}
\label{app:majority_proof}

As stated in \Cref{sec:maj}, we perform a reduction from \sat{} to \netlearn{}, adapting the proof of \Cref{thm:nphardness_bayes}.
The graph construction we use is described in \Cref{ssec:maj_graph}.

\ourparagraph{Literal notation}
We now introduce some notation for use in subsequent sections.
Let $ \ell $ be a literal of some variable $ x \in \vars $, meaning either $ \ell = x $ or $ \ell = \lnot x $.
We extend the notation of the cell to be $ \cell \ell \deq \cell x $.
We say that cell $\cell x$ is ``on'' under a decision ordering $\sigma$ if $\sigma(\lnot x_i) < \sigma(x_i)$; otherwise, cell $\cell x$ is ``off''.
We say that a \emph{literal} $ \ell $ is \emph{on} if $ \ell = x $ and $ \cell x $ is on, or $ \ell = \lnot x $ and $ \cell x $ is off; otherwise, literal $ \ell $ is off.

\subsection{Gadget Learning Rates}
We begin by examining the learning rate of an arbitrary cell under a pair of orderings in which the cell is either ``on'' or ``off''. The following lemma shows that cells achieve the same learning rate under either of these orderings, and that this learning rate is the best possible over all orderings. 

\begin{lemma}[Majority Dynamics Cell LR] \label{lemma:MD_cellLearningRate}
    Let $x \in \vars$.
    Let $ q = \frac 12 $, and $ p $ be given.
    Then \[
	\oclr (\cell x) = 2p + 3p^2 -2p^3.
    \]
\end{lemma}

\begin{proof}
    First, observe that for the optimal ordering $ \sigma^* $, it is always beneficial to put the dummy node $ d_x $ \emph{before} the nodes $ x $ and $ \lnot x $.
    This is because $ d_x $ itself has no incoming edges, so it cannot benefit from any further information, and it has edges going to $ x $ and $ \lnot x $, which can only increase their chances of getting the correct answer.
    Hence, we can see that $ \oclr(d_x) = p $.

    The case of the remaining two nodes is symmetric, WLOG let us assume that $ \sigma^*(x) < \sigma^*(\lnot x) $ (so the cell $ \cell x $ is ``off'').
    The node $ x $ then receives, apart from its private information, the action of $ d_x $.
    Due to tie-breaking, it still always chooses to believe its private signal, so $ \oclr (x) = p $.

    Finally, the node $ \lnot x $ receives the private signal, $ s_{\lnot x} $, the action of $ d_x $, $ a_{d_x} $ and the action of $ x $, $ a_x $.
    Notice that all these three pieces of information are independent (since $ x $ was not influenced by $ d_x $ due to tie-breaking).
    This makes it easy to compute the learning rate of node $ \lnot x $ as the probability that at least two of the three pieces of information are correct \begin{align*}
        \oclr(\lnot x) &= \prt{at least two of $ s_{\lnot x}, a_x, a_{d_x} $ are correct} \\
        &= 3p^2 - 2p^3.
    \end{align*}

    By linearity of expectation, it then holds that \[
	\oclr(\cell x) = 2p + 3p^2 -2p^3.
    \]
\end{proof}

We defer the question of which of the two orderings is optimal to the next section. Addressing this first requires examining the learning rates of the clause gadgets, which are also affected by cell node orderings. We begin with a general expression for the CLR of any clause gadget under a particular partial ordering. We then use this expression to compute the CLR for specific clauses, observing that the optimal clause node ordering must respect the partial ordering given earlier.

\begin{lemma}[Majority Clause LR] \label{lemma:MD_genOrderingCLR}
    For any clause $ C = \left( \ell_1 \lor \ell_2 \lor \ell_3 \right) $, let $ a,b,c $ be the literal nodes in the clause gadget corresponding to $\ell_1$, $\ell_2$, and $\ell_3$. Also, let $ \alpha, \beta, \gamma $ be the cell variable nodes corresponding to $\ell_1$, $\ell_2$, and $\ell_3$ (see \Cref{fig:alphabetagamma} for details of the construction). Suppose that $\alpha$, $\beta$, and $\gamma$ correctly predict the ground truth with probabilities $p_\alpha$, $p_\beta$, and $p_\gamma$, respectively.
    Consider a partial ordering $ \sigma $ which orders all variable and dummy nodes before all literal nodes. WLOG, suppose $ \sigma $ also orders $a$ before $b$ before $c$.
    Then the cumulative learning rate of the clause gadget under $\sigma$ is \begin{align*}
	\clr(\gadget C, \sigma) = 2p + \Pr[a] + \Pr[b] + \Pr[c], 
    \end{align*}
    where 
    \begin{align*}
        \Pr[a] &= p((1-p)p_\alpha + p(p+2)),\\
        \Pr[b] &= p^2 \left((1-p)^2 (p_\alpha + 5{p_\beta})+p (4-3 p)\right),\\
        \Pr[c] &= p^2 \left((1-p)^2 (p_\alpha(1-{p_\gamma} ({p_\beta}+p-1)+{p_\beta} p)+ 2p_\beta p -{p_\gamma}({p_\beta} (5 p-3)-3 p-2)) + p(4-3p)\right).
    \end{align*}
\end{lemma}

\begin{proof}
     Note that the dummy nodes have no incoming edges, and therefore each contributes a learning rate of $p$. We now compute the learning rates of $ a $, $ b $, and $ c $.
    
     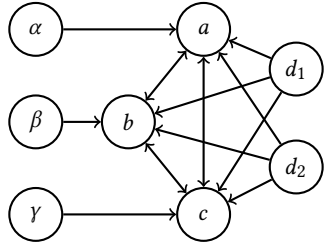
\begin{figure}[t!]
	\centering
	\begin{tikzpicture}[
		node distance = 5mm and 5mm,
		xsh/.style = { xshift=10mm },
	stff/.style={circle, draw=black, \figThickness, minimum size=\figCircSize, inner sep=0pt},
	]
		\node[stff]        (alpha)                  {$ \alpha $};
		\node[stff]        (beta)    [below=of alpha]  {$ \beta $};
		\node[stff]        (gamma)    [below=of beta]  {$ \gamma $};

		\node[stff]        (a)    [right=of alpha,xsh]  {$ a $};
		\node[stff]        (b)    [right=of beta]  {$ b $};
		\node[stff]        (c)    [right=of gamma,xsh]  {$ c $};
		\node[stff]        (d1)    [right=of b, xsh,yshift=20]  {$ d_1 $};
		\node[stff]        (d2)    [right=of b, xsh,yshift=-6mm]  {$ d_2 $};

		\draw[->, \figThickness] (alpha)  to (a);
		\draw[->, \figThickness] (beta)  to (b);
		\draw[->, \figThickness] (gamma)  to (c);

		\draw[<-, \figThickness] (a)  to  (d1);
		\draw[<-, \figThickness] (b)  to  (d1);
		\draw[<-, \figThickness] (c)  to  (d1);
		\draw[<-, \figThickness] (a)  to  (d2);
		\draw[<-, \figThickness] (b)  to  (d2);
		\draw[<-, \figThickness] (c)  to  (d2);
		\draw[<->, \figThickness] (a)  to (c);
		\draw[<->, \figThickness] (b)  to (c);
		\draw[<->, \figThickness] (b)  to (a);
	\end{tikzpicture}
	\caption{Construction used in General Ordering Learning Rate proof. $\alpha, \beta$ and $\gamma$ are the vertices from the cell gadgets corresponding to the literals from $C$.}
 \Description{A graph of the Majority gadget  with literals $a,b,c$, along with $alpha,beta,gamma$, from which edges go only to $a,b,c,$ respectively.}
     \label{fig:alphabetagamma}
\end{figure}

    The first node, $a$, is correct iff \begin{enumerate}[ ]
        \item both $d_1$ and $d_2$ are correct, and at least one of their private signal and $\alpha$ is correct: \[
        \Pr[a \mid d_1, d_2] = 1 - (1-p)(1-p_\alpha),
        \]
        \item or one of $d_1$ and $d_2$ are correct, and their private signal is correct: \[
        \Pr[a \mid d_1, \bar d_2] = \Pr[a \mid \bar d_1, d_2] = p,
        \]
        \item or neither $d_1$ nor $d_2$ are correct, but both $\alpha$ and their private signal are correct: \[
        \Pr[a \mid \bar d_1, \bar d_2] = pp_\alpha.
        \]
	\end{enumerate}

	Therefore, the total probability of $a$ being correct is \begin{align*}
            \Pr[a] &= p^2 \left( 1-(1-p)(1-p_\alpha) \right) + 2p(1-p)p + (1-p)^2pp_\alpha = p (p_\alpha- p(p_\alpha -p-2))= p((1-p)p_\alpha + p(p+2)).
	\end{align*}
    
    Node $b$ is correct with conditional probabilities:
    \begin{align*}
        \Pr[b \mid d_1, d_2] &= 1 - (1-p_\beta)(1-\Pr[a \mid d_1, d_2])(1-p), \\
        \Pr[b \mid d_1, \bar d_2] &= p_\beta\Pr[a \mid d_1, \bar d_2]p + (1-p_\beta)\Pr[a \mid d_1, \bar d_2]p + p_\beta(1-\Pr[a \mid d_1, \bar d_2])p + p_\beta\Pr[a \mid d_1, \bar d_2](1-p) \\
                                  & = \Pr[b \mid \bar d_1, d_2], \\
        \Pr[b \mid \bar d_1, \bar d_2] &= p_\beta\Pr[a \mid \bar d_1, \bar d_2]p.
    \end{align*}

    The total probability of $b$ being correct is then \[
	\pr{b} = p^2 \left((1-p)^2 (p_\alpha + 5{p_\beta})+p (4-3 p)\right).
    \]

    We can also compute the following joint probabilities between $a$ and $b$, conditioning again on the dummy nodes: 
    \begin{align*}
        \Pr[a,b \mid d_1, d_2] &= \Pr[a \mid d_1, d_2],\\
        \Pr[\bar a, \bar b \mid d_1, d_2] &= (1-p_\beta)(1-p)(1-\Pr[a \mid d_1, d_2]),\\
        \Pr[a,b \mid d_1, \bar d_2] &= (1 - (1-p_\beta)(1-p)) \Pr[a \mid d_1, \bar d_2],\\
        \Pr[\bar a, \bar b \mid d_1, \bar d_2] &= (1-p_\beta p)(1-\Pr[a \mid d_1, \bar d_2]),\\
        \Pr[a,b \mid \bar d_1, \bar d_2] &= p_\beta p \Pr[a \mid \bar d_1, \bar d_2],\\
        \Pr[\bar a, \bar b \mid \bar d_1, \bar d_2] &= 1-\Pr[a \mid \bar d_1, \bar d_2].
    \end{align*}
    
    Similarly, $c$ is correct with conditional probabilities 
    \begin{align*}
        \Pr[c \mid d_1, d_2] &= p + (1-p)(p_\gamma(1-\Pr[\bar a, \bar b \mid d_1, d_2]) + (1-p_\gamma)\Pr[a, b \mid d_1, d_2]), \\
        \Pr[c \mid d_1, \bar d_2] &= p(1-(1-p_\gamma)\Pr[\bar a, \bar b \mid d_1, \bar d_2]) + (1-p) p_\gamma \Pr[a, b \mid d_1, \bar d_2] \\
        & = \Pr[c \mid \bar d_1, d_2], \\
        \Pr[c \mid \bar d_1, \bar d_2] &= p(p_\gamma(1-\Pr[\bar a, \bar b \mid \bar d_1, \bar d_2]) + (1-p_\gamma)\Pr[a, b \mid \bar d_1, \bar d_2]).
    \end{align*}
    
    Therefore, the total probability of $c$ being correct is: \[
	\pr{c} = p^2 \left((1-p)^2 (p_\alpha(1-{p_\gamma} ({p_\beta}+p-1)+{p_\beta} p)+ 2p_\beta p -{p_\gamma}({p_\beta} (5 p-3)-3 p-2)) + p(4-3p)\right).
    \]
    
    Finally, nodes $ d_1 $ and $ d_2 $ are both correct with probability $ p $.
    So the total learning rate of the clause is \begin{align*}
	2p + \pr{a} + \pr{b} + \pr{c}.
    \end{align*}
\end{proof}

\begin{lemma}[$1 \lor 0 \lor 0$ Learning Rate]\label{lemma:100CLR} 
    Let $C=\ell_1 \lor \ell_2 \lor \ell_3$ be a clause.
    Suppose that an optimal ordering $\sigma^*$ sets the literal $\ell_1$ to be ``on'', and the rest of the literals to be ``off''.
    Then, in the gadget for $C$, $\sigma^*$ places $\ell_1$ first, then the other two, and finally the node $C$.
    Further, the resulting CLR is \[
        \pone = p \left(2 + 2 p + 6 p^2 + 11 p^3 + 4 p^4 - 51 p^5 - 6 p^6 + 21 p^7 + 
   115 p^8 - 136 p^9 + 13 p^{10} + 36 p^{11} - 12 p^{12} \right).
    \]
\end{lemma}

\begin{proof}
    Since $\sigma^*$ is an optimal ordering, it must place the dummy nodes $d_1$ and $d_2$, as well as all variable nodes, before all literal nodes. Otherwise, the literal nodes receive strictly less information.
    To determine the order in which the literal nodes $\ell_1, \ell_2, \ell_3$ are placed, we can simply compute the expected learning rate for every permutation of them, and then pick the maximizing order.
    Note that $\ell_2$ and $\ell_3$ are identical, so we can treat permutations which just exchange these two as identical as well.

    We use the learning rate for the general clause gadget, as computed in \Cref{lemma:MD_genOrderingCLR}.
    We simply analyze the three non-isomorphic cases: $\ell_1=a, \ell_1=b,$ and $\ell_1=c$. Note that the first case corresponds to setting $p_\alpha = 3p^2 - 2p^3$ and $p_\beta = p_\gamma = p$, and similarly for the latter cases. We get the following cumulative learning rates: \begin{align*}
        (\ell_1=a) &\to p \left(2 + 2 p + 6 p^2 + 11 p^3 + 4 p^4 - 51 p^5 - 6 p^6 + 21 p^7 + 115 p^8 - 136 p^9 + 13 p^{10} + 36 p^{11} - 12 p^{12} \right),\\
        (\ell_1=b) &\to -p \left(-2 - 3 p + p^2 - 28 p^3 + 6 p^4 + 85 p^5 - 95 p^6 + 132 p^7 - 248 p^8 + 165 p^9 + 42 p^{10} - 84 p^{11} + 24 p^{12}\right),\\
        (\ell_1=c) &\to p \left(2 + 3 p + 2 p^2 + 21 p^3 - 25 p^4 - 29 p^5 + 143 p^6 - 440 p^7 + 684 p^8 - 458 p^9 + 54 p^{10} + 72 p^{11} - 24 p^{12}\right).
    \end{align*}
    
    The first polynomial is always strictly larger than the other two for $\frac 12 < p < 1$, so placing $l_1$ first in the ordering gives the highest learning rate in the clause. So we conclude that 
    \begin{align*}
        \pone = p \left(2 + 2 p + 6 p^2 + 11 p^3 + 4 p^4 - 51 p^5 - 6 p^6 + 21 p^7 + 
   115 p^8 - 136 p^9 + 13 p^{10} + 36 p^{11} - 12 p^{12} \right).
    \end{align*}
    Moving forward, we refer to this probability simply by $\pone$, regardless of the order that the literals appear in the clause itself.
\end{proof}

Before we continue with other clauses, note that the CLR of any clause increases with the number of true literals in the clause. One can verify this by checking that for any literal node or the clause node $C$ in the gadget, its probability of correctly predicting the ground truth increases monotonically with each of the variable probabilities. So we will only compute the CLR for clauses which involve the smallest and largest possible variable probabilities.

\begin{lemma}[{$0 \lor 0 \lor 0$, $1 \lor 1 \lor 1$ Learning Rates}]
    \label{lemma:000/111CLR}
    Let $C=\ell_1 \lor \ell_2 \lor \ell_3$ be a clause.
    Suppose that an optimal ordering $\sigma^*$ sets all three literals to false.
    Then the resulting CLR is \[
        \pzero = p \left(2 + 3 p + 2 p^2 + 25 p^3 - 42 p^4 + 23 p^5 - 67 p^6 + 102 p^7 - 
   31 p^8 - 24 p^9 + 12 p^{10}\right).
    \]
    If instead $\sigma^*$ sets all three literals to true, then the resulting CLR is  
    \begin{align*}
     \pthree = p \left(2 + 2 p + 3 p^2 + 14 p^3 + 22 p^4 - 66 p^5 - 69 p^6 + 310 p^7 - 688 p^8 + 710 p^9 + \right.\\
     \left.756 p^{10} - 2581 p^{11} + 2304 p^{12} - 558 p^{13} - 
   372 p^{14} + 264 p^{15} - 48 p^{16} \right).
    \end{align*}
\end{lemma}
\begin{proof}
    As before, the optimal ordering $\sigma^*$ must place node $C$ after all literal nodes. Observe that by symmetry, the ordering of $l_1$, $l_2$, and $l_3$ doesn't matter, since all three variable probabilities are identical. Using the \Cref{lemma:MD_genOrderingCLR} and supplying variable probabilities $p_\alpha = p_\beta = p_\gamma = p$, we get the desired CLR $p^000$.
    Supplying variable probabilities $p_\alpha = p_\beta = p_\gamma = 3p^2 - 2p^3$ gives the desired value of $\pthree$.
\end{proof}

\begin{lemma}[$ 0 \lor 1 \lor 1 $ Learning Rate]
    \label{lemma:110CLR}
    Let $C=\ell_1 \lor \ell_2 \lor \ell_3$ be a clause.
    Suppose that an optimal ordering $\sigma^*$ sets one literal to false, the other literals to true.
    Then the resulting CLR is \[
	\ptwo = p \left(12 p^8-54 p^7+76 p^6-14 p^5-40 p^4+9 p^3+12 p^2+2 p+2\right).
    \]
\end{lemma}

\begin{proof}
    WLOG assume that $ \ell_2 = \ell_3 = 1 $ and $ \ell_1 = 0 $
    We analyze the probabilities of all three (non-isomorphic) possible orderings by setting proper $ p_\alpha, p_\beta, p_\gamma $ from \Cref{lemma:MD_genOrderingCLR}. 
    This results in the following CLRs: \begin{align*}
        (\ell_1=\alpha) &\to p \left(-24 p^{10}+132 p^9-282 p^8+279 p^7-116 p^6+28 p^5-36 p^4+11 p^3+8 p^2+3 p+2\right),\\
        (\ell_1=\beta) &\to -8 p^{10}+44 p^9-86 p^8+57 p^7+18 p^6-21 p^5-11 p^4-7 p^3+15 p^2+2 p+2,\\
        (\ell_1=\gamma) &\to p \left(12 p^8-54 p^7+76 p^6-14 p^5-40 p^4+9 p^3+12 p^2+2 p+2\right).
    \end{align*}
    It is not hard to prove, that for $ p \in \left( \frac 12, 1 \right) $, the biggest value is received when $ \ell_1=\gamma $, and thus its CLR is the optimal CLR.
    This finishes the proof.
\end{proof}

\begin{lemma}[Clause Learning Rate Comparison]
\label{lemma:CLRComp}
	For $ p \in \left( \frac 12, 1 \right) $, it holds \[
	    \pthree \geq \ptwo \geq \pone \geq \pzero.
	\]
\end{lemma}

\begin{proof}
    We know the learning rates of $ \pthree ,\ptwo , \pone , \pzero $ by \Cref{lemma:100CLR,lemma:110CLR,lemma:000/111CLR}.
    The statement then follows by simple algebra.
\end{proof}

\subsection{Optimal Ordering \& Restrictions on p} 
Finally, we can now determine the optimal ordering by examining the learning rates derived in the previous section. First, we define an \emph{assignment-induced ordering} below:

\begin{definition}
    Let $\mathcal{A}: \chi \to \{0,1\}$ be any assignment of values to variables. Define the (partial) \emph{ordering induced by $\mathcal{A}$} as follows: if $\mathcal{A}(x_i)=1$, node $x_i$ is ordered after node $\lnot x_i$ (cell $\cell {x_i}$ is ``on''); otherwise, $x_i$ comes before $\lnot x_i$ (cell $\cell {x_i}$ is ``off'').
\end{definition}

\begin{definition}
    Let $\mathcal{A}: \chi \to \{0,1\}$ be an assignment of values to variables maximizing the number of satisfied clauses. Then $\mathcal{A}$ is a \emph{maximal assignment}.
\end{definition}

\begin{lemma}[Optimal Ordering] \label{lemma:bestOrder}
    Let $\mathcal{A}^*$ be a maximal assignment.  Let $p(M) < 1$ be a threshold probability determined by $M$, the number of clauses. Then for all $p \geq p(M)$, the decision ordering $\sigma^*$ which places all dummy nodes first, then all variable nodes respecting the partial ordering induced by $\mathcal{A}^*$, then all literal nodes, and finally all clause nodes, maximizes the network learning rate.
\end{lemma}
\begin{proof}
    First note that by the same reasoning as in Lemma \ref{lemma:MD_cellLearningRate}, there is an optimal ordering $\sigma^*$ which places all dummy nodes before all variable nodes. Similarly, the same reasoning as in Lemma \ref{lemma:100CLR} gives that there is an optimal $\sigma^*$ in which all literal nodes are ordered after all variable nodes, and all clause nodes are ordered after all literal nodes. So all that remains is to show that there is an optimal $\sigma^*$ which respects the ordering induced by $\mathcal{A}^*$. 

    First recall that by Lemma \ref{lemma:MD_cellLearningRate}, an optimal ordering always gives the same cumulative learning rate $p_i^0 = p_i^1$ per cell regardless of whether the cell is ``on'' or ``off''. So the ordering of variable nodes only affects the network learning rate through the clauses. Consider any other assignment $\mathcal{A}$ which is not maximal. Let $S^*$ be the number of satisfied clauses under $\mathcal{A}^*$ and $S < S^*$ that under $\mathcal{A}$. By \Cref{lemma:100CLR,lemma:000/111CLR}, the CLR for any satisfied clause is lower bounded by $\pone$, which is strictly greater than that for an unsatisfied clause, $\pzero$. So having more satisfied clauses should improve the CLR. However, note also that $\pthree \geq \ptwo \geq \pone \geq \pzero$, by \Cref{lemma:110CLR}. So in the most extreme case, all $S$ satisfied clauses under $\mathcal{A}$ are satisfied by having three true literals, while all $S^*$ satisfied clauses under $\mathcal{A}^*$ are satisfied by having one true literal. 

    Consider that extreme case, and further impose the worst-case choices of $S$ and $S^*$ by setting $S^* = M$ and $S = M-1$. Then over all clause gadgets, the ordering induced by $\mathcal{A}^*$ gives a CLR of $M \pone$, while the ordering induced by $\mathcal{A}$ gives a CLR of $(M-1) \pthree + \pzero$. In order for the ordering induced by $\mathcal{A}^*$ to maximize the network learning rate, we must have
    \begin{align*}
        M \pone &\geq (M-1) \pthree + \pzero\\
        \frac{\pone - \pzero}{\pthree - \pone} &\geq M-1
    \end{align*}
    One can verify that the left-hand side can be lower bounded by $\frac{1}{6-6p}$ for all $p \in (0.5,1)$, so the following is sufficient:
    \begin{align*}
        \frac{1}{6-6p} \geq M-1 \Longrightarrow p \geq p(M) \deq \frac{6M-7}{6M-6},
    \end{align*}
    which is well-defined for any $M \geq 2$. Thus, the $\sigma^*$ respecting the ordering induced by $\mathcal{A}^*$ is optimal for all $p \geq p(M)$.
\end{proof}

To recap, given any instance $\varphi$ of \sat{} with $N$ variables and $M$ clauses, we can construct a network $G_\varphi$ with ground truth prior $q = 1/2$ and signal accuracy $p \geq p(M)$. In particular, whenever $\varphi$ is satisfiable under some maximal assignment $\mathcal{A}^*$, there is an optimal decision ordering $\sigma^*$ which respects the ordering induced by $\mathcal{A}^*$, and which achieves a network CLR of at least $N(2p+3p^2-2p^3) + M\pone$. Otherwise, if $\varphi$ is non-satisfiable under any maximal assignment $\mathcal{A}^*$, then the optimal decision ordering $\sigma^*$ respecting the ordering induced by $\mathcal{A}^*$ achieves a network CLR of no more than $N(2p+3p^2-2p^3) + (M-1)\pthree + \pzero$. Choosing $p \geq p(M)$ allowed us to show the optimality of $\sigma^*$, as well as to separate the learning rates in networks corresponding to satisfiable and non-satisfiable formulas. All that remains to complete the reduction is to pick an appropriate choice of $\varepsilon$, such that a formula graph $G_\varphi$ achieves expected network learning rate greater than $\varepsilon$ under an optimal ordering iff $\varphi$ is satisfiable.

\subsection{Picking the Epsilon}
We will simply pick $\varepsilon$ to lie exactly halfway between the satisfiable and non-satisfiable network learning rates. Recalling that each variable gadget contains $3$ vertices and each clause gadget contains $4$ vertices, we can compute the learning rates below. For networks corresponding to satisfiable formulas, we have at least
\begin{align*}
    \frac{N\pcell + M\pone}{3N + 5M},
\end{align*}
and for networks corresponding to non-satisfiable formulas, we have at most
\begin{align*}
    \frac{N\pcell + (M-1)\pthree + \pzero}{3N + 5M}.
\end{align*}
This gives 
\begin{align*}
    \varepsilon = \frac{1}{2}\left(\frac{2N\pcell + M\pone + (M-1)\pthree + \pzero}{3N + 5M}\right),
\end{align*}
thus completing the reduction.


\end{document}